\documentclass[11pt]{article}

\usepackage{natbib}
\usepackage{amssymb,amsmath,amsthm,latexsym,amsfonts,amscd,dsfont,enumerate}
\usepackage{graphicx, xcolor}
\usepackage{bm}
\usepackage[T1]{fontenc}
\usepackage{lmodern}
\usepackage{fixltx2e} 
\usepackage[top=1.2in, left=0.9in, bottom=1.2in, right=0.9in]{geometry}

\usepackage{tikz}
\usetikzlibrary{arrows,positioning}
\tikzset{
    >=stealth',
    punkt/.style={
           rectangle,
           rounded corners,
           draw=black, very thick,
           text width=6.5em,
           minimum height=2em,
           text centered},
    pil/.style={
           ->,
           thick,
           shorten <=2pt,
        shorten >=2pt,
   }
}

\usepackage[]{hyperref}
 \hypersetup{
 colorlinks=true, 
 breaklinks=true, 
 urlcolor= blue, 
 linkcolor= blue, 
 bookmarksopen=true, 
 pdfauthor={Bichuch, M., Capponi, A., Sturm, S.}
 }

\newcommand{\G}{\mathcal{G}}
\newcommand{\gt}{\mathcal{G}_t}

\newcommand{\R}{\mathbb{R}}

\newcommand{\Px}{\mathbb{P}}
\newcommand{\Exx}{\mathbb{E}}
\newcommand{\E}{\Exx}
\newcommand{\Qxx}{\mathbb{Q}}

\newcommand{\Gx}{\mathbb{G}}

\newcommand{\vv}{{\bar{{v}}}}

\def\ind{{\mathchoice{1\mskip-4mu\mathrm l}{1\mskip-4mu\mathrm l}
{1\mskip-4.5mu\mathrm l}{1\mskip-5mu\mathrm l}}}

\newcommand{\rrp}{r_r^+}
\newcommand{\rrm}{r_r^-}
\newcommand{\rcp}{r_c^+}
\newcommand{\rcm}{r_c^-}
\newcommand{\rfp}{r_f^+}
\newcommand{\rfm}{r_f^-}

\newcommand{\XVA}{\mbox{XVA}}
\newcommand{\CVA}{\mbox{CVA}}
\newcommand{\DVA}{\mbox{DVA}}

\newcommand{\hIQ}{h_I^{\Qxx}}
\newcommand{\hCQ}{h_C^{\Qxx}}

\newtheorem{theorem}{Theorem}[section]
\newtheorem{definition}[theorem]{Definition}

\newtheorem{proposition}[theorem]{Proposition}
\newtheorem{remark}[theorem]{Remark}
\newtheorem{lemma}[theorem]{Lemma}

\newtheorem{assumption}[theorem]{Assumption}

\newcommand{\psir}{\psi^r}

\title{Arbitrage-Free Pricing of XVA -- Part I:\\ Framework and Explicit Examples \footnote{The paper \cite{BCM} subsumes the present working paper as well as \cite{BicCapSturm}. } }

\author{
Maxim Bichuch \thanks{Email: mbichuch@wpi.edu, Department of Mathematical Sciences, Worcester Polytechnic Institute}
\and
Agostino Capponi \thanks{Email: ac3827@columbia.edu, Industrial Engineering and Operations Research Department, Columbia University}
\and
Stephan Sturm \thanks{Email: ssturm@wpi.edu, Department of Mathematical Sciences, Worcester Polytechnic Institute}
}

\begin{document}

\maketitle

\begin{abstract}
We develop a framework for computing the total valuation adjustment (XVA) of a European claim accounting for funding costs, counterparty credit risk, and collateralization. Based on no-arbitrage arguments, we derive
nonlinear backward stochastic differential equations (BSDEs) associated with the replicating portfolios of long and short positions in the claim. This leads to the definition of buyer{'}s and seller{'}s XVA, which in turn
identify a no-arbitrage interval. In the case that borrowing and lending rates coincide, we provide a fully explicit expression for the uniquely determined XVA, expressed as a percentage of the price of the traded claim,
and for the corresponding replication strategies. This extends the result of \cite{Piterbarg} by incorporating the effect of premature contract termination due to default risk of the trader and of his counterparty.
\end{abstract}

\vspace{5mm}

\begin{flushleft}
\textbf{Keywords:} XVA, counterparty credit risk, funding spreads, backward stochastic differential equations, arbitrage-free pricing. 
\end{flushleft}

\section{Introduction}
When managing a portfolio, a trader needs to raise cash in order to finance a number of operations. Those include maintaining the hedge of the position, posting collateral resources, and paying interest on collateral received. Moreover, the trader needs to account for the possibility that the position may be liquidated prematurely due to his own or counterparty{'}s default, hence entailing additional costs due to the closeout procedure. Cash resources are provided to the trader by his treasury
desk, and must be remunerated. If he is borrowing, then he will be charged an interest rate depending on current market conditions as well as on his own credit quality. Such a rate is usually higher than the lending rate at which the trader would lend excess cash proceeds from his investment strategy to the treasury. The difference between borrowing and lending rate is also referred to as \textit{funding spread}.

Even though pricing by replication can still be put to work under this rate asymmetry, the classical Black-Scholes formula no longer yields the price of the claim. In the absence
of default risk, few studies have been devoted to pricing and hedging claims in markets with differential rates. \cite{Korn95} considers option pricing in a market with a higher borrowing than lending rate, and derives an interval of acceptable prices for both the buyer and the seller. \cite{Cvi93} consider the problem of hedging contingent claims under portfolio constraints allowing for a higher borrowing than lending rate. \cite{ElKaroui} study the super-hedging price of a contingent claim under rate asymmetry via nonlinear backward stochastic differential equations (BSDEs).

The above studies do not consider the impact of counterparty credit risk on valuation and hedging of the derivative security. The new set of rules mandated by the Basel Committee (\cite{Basel3}) to govern bilateral trading in OTC markets requires to take into account default and funding costs when marking to market derivatives positions. This has originated a growing stream of literature, some of which is surveyed next. \cite{Crepeya} and \cite{Crepeyb} introduce a BSDE approach for the valuation of counterparty credit risk taking funding constraints into account. He decomposes the value of the transaction into three separate components, the collateral, the hedging assets used to hedge market risk of the portfolio as well as
counterparty credit risk, and the funding assets needed to finance the hedging strategy. An accompanying numerical study was conducted in \cite{CrepeyNum}. Along similar lines, \cite{BrigoPalCCP} first derive a risk-neutral pricing formula for a contract taking into account counterparty credit risk, funding, and collateral servicing costs, and then provide the corresponding BSDE representation. \cite{Piterbarg} derives a closed form solution for the price of a derivative contract in presence of funding costs and collateralization, but ignores the possibility of counterparty{'}s default. Moreover, he assumes that borrowing and lending rates are equal, an assumption that is later relaxed by \cite{Mercurio}. \cite{Burgard}
 and \cite{BurgardCR} generalize \cite{Piterbarg}{'}s model to include default risk of the trader and of his counterparty. They derive PDE representations for the price of the derivative via a replication approach, assuming the absence of arbitrage and sufficient smoothness of the derivative price.
\cite{br} develop a general semimartingale market framework and derive the BSDE representation of the wealth process associated with a self-financing trading strategy that replicates a default-free claim. As in \cite{Piterbarg} they do not take counterparty credit risk into account. \cite{br} consider the viewpoint of the hedger, but do not derive arbitrage bounds for unilateral prices. These are instead analyzed in \cite{NiRut}, who also examine the existence of fair bilateral prices. A good overview of the current literature is given in \cite{CrepeyBieleckiBrigo}.

In the present article we introduce a valuation framework which allows us to quantify the total valuation adjustment, abbreviated as XVA, of a European type claim. We consider an underlying portfolio consisting of a default-free stock and two risky bonds underwritten by the trader{'s} firm and his counterparty. Stock purchases and sales are financed through the security lending market. We allow for asymmetry between treasury borrowing and lending rates, repo lending and borrowing rates, as well as between interest rates paid by the collateral taker and received by the collateral provider.

We derive the nonlinear BSDEs associated with the portfolios replicating long and short positions in the traded claim, taking into account counterparty credit risk and closeout payoffs exchanged at default.  This extends, in our specific setting, the general semi-martingale framework introduced by \cite{br} by including counterparty default risk. Due to rate asymmetries, the BSDE which represents the valuation process of the portfolio replicating a long position in the claim cannot be directly obtained (via a sign change) from the one replicating a short position. More specifically, there is a no-arbitrage interval which can be defined in terms of the buyer{'}s and the seller{'}s XVA.

When borrowing and lending rates are the same, buyer{'}s and seller{'}s XVA coincide, and we can develop a fully explicit expression which precisely identifies the contributions from funding costs, credit valuation adjustment (CVA), and debit valuation adjustment (DVA). Moreover, we can express the total valuation adjustment as a percentage of the publicly available price of the claim. This gives an interpretation of the XVA in terms of the cost of a trade, and has risk management implications because it pushes banks to attempt to reduce their borrowing costs in addition to choosing trades that minimize funding costs.

We consider two different specializations of our framework. The first recovers the framework put forward by \cite{Piterbarg}. Here, we also give the exact expression for the number of shares associated with the replicating strategy of the investor. The latter is shown to consist of two components. The first is a funding adjusted delta hedging strategy, and the second is a correction based on the gap between funding and collateral rates. The second specialization is an extension of \cite{Piterbarg}{'}s model, which also incorporates the possibility of counterparty credit risk and hence includes trader and counterparty bonds into the portfolio of replicating securities. By means of cumbersome computations, we provide an explicit decomposition of XVA into three main components: funding adjusted price of the transaction exclusive of collateral and default costs, funding adjusted closeout payments, and funding costs of the collateralization procedure.

The paper is organized as follows. We develop the model in Section \ref{sec:model} and introduce replicated claim and collateral process in Section \ref{sec:claim}. We analyze arbitrage-free pricing of XVA in Section \ref{sec:BSDEform}, and show that credit valuation adjustments can be recovered as special cases of our general formula. Section \ref{sec:expexp} develops an explicit expression for the XVA under equal borrowing and lending rates. Section \ref{sec:conclusions} concludes the paper. Some proofs of technical results are delegated to the Appendix \ref{App_BSDE}.

\section{The model}\label{sec:model}

We consider a probability space $(\Omega,\G,{\Px})$ rich enough to support all subsequent constructions. Here, $\Px$ denotes the physical probability measure. Throughout the paper, we refer to ``$I$'' as the investor, trader or hedger interested in computing the total valuation adjustment, and to ``$C$'' as the counterparty to the investor in the transaction. The background or reference filtration that includes all market information except for default events and augmented by all $(\mathcal{G},\Px)$-nullsets, is denoted by $\mathbb{F} := (\mathcal{F}_t)_{t \geq 0}$. The filtration containing default event information is denoted by $\mathbb{H} := (\mathcal{H}_t)_{t \geq 0}$. Both filtrations will be specified in the sequel of the paper. We denote by ${\Gx}:=(\gt)_{t \geq 0}$ the enlarged filtration given by $\gt := \mathcal{F}_t \vee \mathcal{H}_t$, augmented by its nullsets. Note that because of the augmentation of $\mathbb{F}$  by nullsets, the filtration $(\G_t)$ satisfies the usual conditions of completeness and right continuity; see Section 2.4 of \cite{Belanger}.

We distinguish between \textit{universal} instruments, and \textit{investor specific} instruments, depending on whether their valuation is \textit{public} or \textit{private}. Private valuations are based on discount rates, which depend on investor specific characteristics, while public valuations depend on publicly available discount factors. Throughout the paper, we will use the superscript $\wedge$ when referring specifically to public valuations. Section \ref{sec:univ} introduces the universal securities. Investor specific securities are introduced in Section \ref{sec:hedgerspe}.

\subsection{Universal instruments} \label{sec:univ}
This class includes the default-free stock security on which the financial claim is written, and the security account used to support purchases or sales of the stock security. Moreover, it includes the risky bond issued by the trader as well as the one issued by his counterparty.

\paragraph{The stock security.}
We let $\mathbb{F} := (\mathcal{F}_t)_{t \geq 0}$ be the $(\mathcal{G},\Px)$-augmentation of the filtration generated by a standard Brownian motion $W^{\Px}$ under the measure $\Px$. Under the physical measure, the dynamics of the stock price is given by
\[
    dS_t = \mu S_t \,dt + \sigma S_t \,dW_t^{\Px},
\]
where $\mu$ and $\sigma$ are constants denoting, respectively, the appreciation rate and the volatility of the stock.

\paragraph{The security account.}
Borrowing and lending activities related to the stock security happen through the security lending or repo market. We do not distinguish between security lending and repo, but refer to all of them as repo transactions. We consider two types of repo transactions: security driven and cash driven, see also \cite{Adrian}. The security driven transaction is used to overcome the prohibition on ``naked'' short sales of stocks, that is the prohibition to the trader of selling a stock which he does not hold and hence cannot deliver. The repo market helps to overcome this by allowing the trader to lend cash to the participants in the repo market who would post the stock as a collateral to the trader. The trader would later return the stock collateral in exchange of a pre-specified amount, usually slightly higher than the original loan amount. Hence, effectively this collateralized loan has a rate, referred to as the repo rate. The cash lender can sell the stock on an exchange, and later, at the maturity of the repo contract, buy it back and return it to the cash borrower. We illustrate the mechanics of the security driven transaction in Figure \ref{fig:secdriven}.

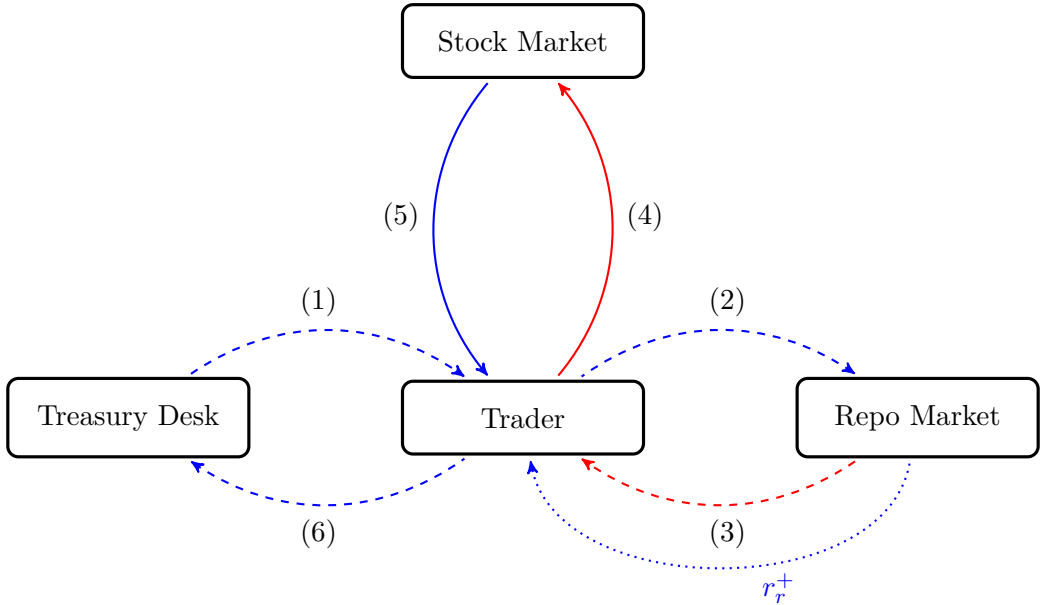
\begin{figure}[ht]
    \centering
    \begin{tikzpicture}[thick,scale=1, every node/.style={transform shape}]
        \node[punkt, inner sep=10pt] (trader) {Trader};
        \node[punkt, inner sep=10pt,  left=2cm of trader] (funder) {Treasury Desk}
            edge[pil, bend left=35, blue, dashed] (trader)
            edge[pil, <-, bend right=35, blue, dashed] (trader);
        \node[above left =1cm of trader] (one) {(1)};
        \node[below left =1cm of trader] (six) {(6)};
        \node[punkt, inner sep=10pt,  above=4cm of trader] (stock) {Stock Market}
            edge[pil, bend right=40, blue] (trader)
            edge[pil, <-, bend left=40, red] (trader);
        \node[below = 2cm of stock.west] (five) {(5)};
        \node[below =2cm of stock.east] (four) {(4)};
        \node[punkt, inner sep=10pt,  right=2cm of trader] (repo) {Repo Market}
            edge[pil, bend left=35, red, dashed] (trader)
            edge[pil, <-, bend right=35, blue, dashed] (trader)
            edge[pil, bend left=80, blue, dotted] (trader);
        \node[above right =1cm of trader] (two) {(2)};
        \node[below right=1cm of trader] (three) {(3)};
        \node[below right=2cm of trader, blue] (rrp) {$\rrp$};
    \end{tikzpicture}
     \caption{Security driven repo activity: Solid lines are purchases/sales, dashed lines borrowing/lending, dotted lines interest due; blue lines are cash, red lines are stock. The treasury desk lends money to the trader (1) who uses it to lend to the repo market (2) receiving in turn collateral (3). He sells the stock on the market to get effectively into a short position (4) earning cash from the deal (5) which he uses to repay his debt to the funding desk (6). As a cash lender, he receives interest at the rate $\rrp$ from the repo market. There are no interest payments between trader and treasury desk as the payments (1) and (6) cancel each other out.}
\label{fig:secdriven}
\end{figure}

The other type of transaction is cash driven. This is essentially the other side of the trade, and is implemented when the trader wants a long position in the stock security. In this case, he borrows cash from the repo market, uses it to purchase the stock security posted as collateral to the loan, and agrees to repurchase the collateral later at a slightly higher price. The difference between the original price of the collateral and the repurchase price defines the repo rate. As the loan is collateralized, the repo rate will be lower than the rate of an uncollateralized loan. At maturity of the repo contract, when the trader has repurchased the stock collateral from the repo market, he can sell it on the exchange. The details of the cash driven transaction are summarized in Figure \ref{fig:cashdriven}.

\begin{figure}[ht]
    \centering
    \begin{tikzpicture}[thick,scale=1, every node/.style={transform shape}]
        \node[punkt, inner sep=10pt] (trader) {Trader};
        \node[punkt, inner sep=10pt,  left=2cm of trader] (funder) {Treasury Desk}
            edge[pil, bend left=35, blue, dashed] (trader)
            edge[pil, <-, bend right=35, blue, dashed] (trader);
        \node[above left =1cm of trader] (one) {(1)};
        \node[below left =1cm of trader] (six) {(6)};
        \node[punkt, inner sep=10pt,  above=4cm of trader] (stock) {Stock Market}
            edge[pil, <-, bend right=40, blue] (trader)
            edge[pil, bend left=40, red] (trader);
        \node[below = 2cm of stock.west] (two) {(2)};
        \node[below =2cm of stock.east] (three) {(3)};
        \node[punkt, inner sep=10pt,  right=2cm of trader] (repo) {Repo Market}
            edge[pil, bend left=35, blue, dashed] (trader)
            edge[pil, <-, bend right=35, red, dashed] (trader)
            edge[pil, <-, bend left=80, blue, dotted] (trader);
        \node[above right =1cm of trader] (four) {(4)};
        \node[below right=1cm of trader] (five) {(5)};
        \node[below right=2cm of trader, blue] (rrm) {$\rrm$};
    \end{tikzpicture}
       \caption{Cash driven repo activity: Solid lines are purchases/sales, dashed lines borrowing/lending, dotted lines interest due; blue lines are cash, red lines are stock. The treasury desk lends money to the trader (1) who uses it to purchase stock (2) from the stock market (3). He uses the stock as collateral (4) to borrow money from the repo market (5) and uses it to repay his debt to the funding desk (6). The trader has thus to pay interest at the rate $\rrm$ to the repo market. There are no interest payments between trader and treasury desk as the payments (1) and (6) cancel each other out.}
    \label{fig:cashdriven}
\end{figure}
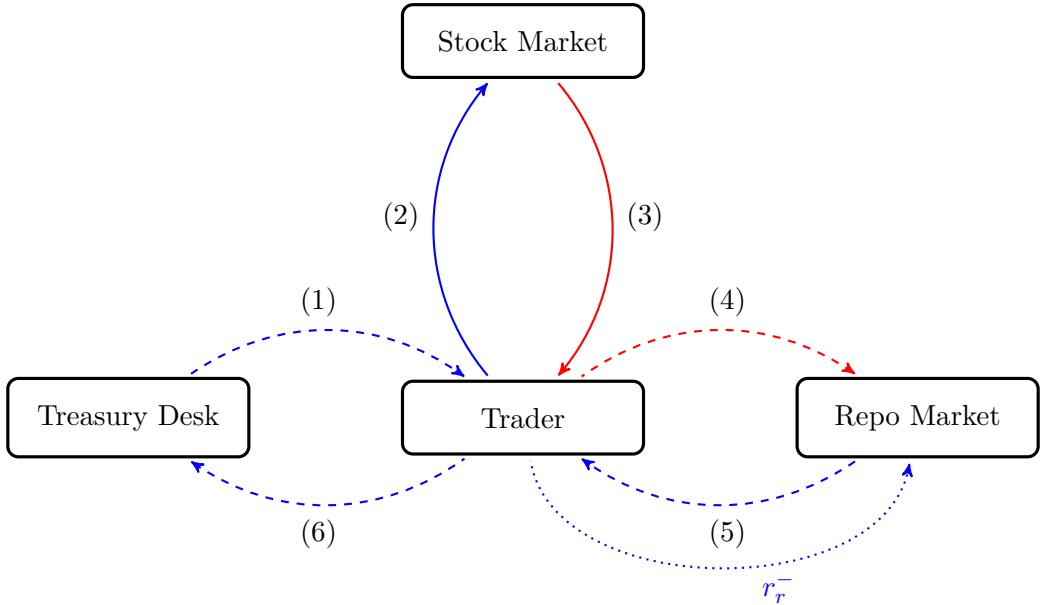

We use $r_r^{+}$ to denote the rate charged by the hedger when he lends money to the repo market and implements his short-selling position. We use $\rrm$ to denote the rate that he is charged when he borrows money from the repo market and implements a long position. We denote by $B^{\rrp}$ and $B^{\rrm}$ the repo accounts whose drifts are given, respectively, by $\rrp$ and $\rrm$. Their dynamics are given by
  \[
    dB_t^{r_r^{\pm}} = r_r^{\pm} B_t^{r_r^{\pm}} dt.
  \]
For future purposes, define
  \begin{equation}\label{eq:Brt}
    B_t^{r_r} := B_t^{r_r}\bigl(\psir \bigr) = e^{\int_0^t r_r (\psir_s) ds},
  \end{equation}
where
  \begin{equation}\label{eq:rr}
    r_r(x) = \rrm \ind_{\{x<0\}}+ \rrp \ind_{\{x>0\}}.
  \end{equation}
Here, $\psir_t$ denotes the number of shares of the repo account held at time $t$. Equations \eqref{eq:Brt}-\eqref{eq:rr} indicate that the trader earns the rate $\rrp$ when lending $\psir>0$ shares of the repo account to implement the short-selling of $-\vartheta$ shares of the stock security, i.e., $\vartheta < 0$. Similarly, he has to pay interest rate $\rrm$ on the $-\psir$ ($\psir < 0$) shares of the repo account that he has borrowed by posting $\vartheta>0$ shares of the stock security as collateral. Because borrowing and lending transactions are fully collateralized, it always holds that
  \begin{equation}\label{eq:selff}
    \psir_t B_t^{r_r} = - \xi_t S_t.
  \end{equation}

\paragraph{The risky bond securities.}
Let $\tau_i$, $i \in \{I, C\}$, be the default times of trader and counterparty. These default times are assumed to be independent exponentially distributed random variables with constant intensities $h_i^{\Px}$, $i \in \{I,C\}$. We use $H_i(t)= \ind_{\{\tau_i\leq t\}}$, $t\geq0$, to denote the default indicator process of $i$. The default event filtration is given by $\mathcal{H}_t = \sigma(H^I_u, H^C_u \; ; u \leq t)$. Such a default model is a special case of the bivariate Cox process framework, for which the $(H)$-hypothesis (see \cite{Elliott}) is well known to hold. In particular, this implies that the $\mathbb{F}$-Brownian motion $W^{\Px}$ is also a $\mathbb{G}$-Brownian motion.

We introduce two risky bond securities underwritten by the trader $I$ and by his counterparty $C$, and maturing at the same time $T$. We denote their price processes by $P^I$ and $P^C$, respectively. For $0 \leq t \leq T$, $i \in \{I, C\}$, the dynamics of their price processes are given by
  \begin{equation}\label{eq:priceproc}
    dP^i_t = (r^i+h_i^{\Px}) P^i_t \, dt - P^i_{t-} \,dH_t^i, \qquad P^i_0 = e^{-(r^i+h_i^{\Px}) T},
  \end{equation}
with return rates $r^i+h_i^{\Px}$, $i \in \{I,C\}$. We do not allow bonds to be traded in the repo market. Our assumption is driven by the consideration that the repurchase agreement market for risky bonds is often illiquid. We also refer to the introductory discussion in \cite{Brennan} stating that even if a bond can be shorted on the repo market, the tenor of the agreement is usually very short.

Throughout the paper, we use $\tau := \tau_I \wedge \tau_C \wedge T$ to denote the earliest of the transaction maturity $T$, trader and counterparty default time.

\subsection{Hedger specific instruments} \label{sec:hedgerspe}

This class includes the funding account and the collateral account of the hedger.

\paragraph{Funding account.}
We assume that the trader lends and borrows moneys from his treasury at possibly different rates. Denote by $\rfp$ the rate at which the hedger lends to the treasury, and by $\rfm$ the rate at which he borrows from it. We denote by $B^{r_f^{\pm}}$ the cash accounts corresponding to these funding rates, whose dynamics are given by
  \[
    dB_t^{r_f^{\pm}} = r_f^{\pm} B_t^{r_f^{\pm}} dt.
  \]
Let $\xi^f_t$ the number of shares of the funding account at time $t$. Define
  \begin{equation}\label{eq:Brf}
    B_t^{r_f}  := B_t^{r_f}\bigl(\xi^f) = e^{\int_0^t r_f(\xi^f_s) ds},
  \end{equation}
where
  \begin{equation}\label{eq:rrf}
    r_f  := r_f(y)= \rfm \ind_{\{{y < 0}\}}+\rfp \ind_{\{{y > 0}\}}.
  \end{equation}
Equations~\eqref{eq:Brf}~-~\eqref{eq:rrf} indicate that if the hedger{'}s position at time $t$, $\xi^f_t$, is negative, then he needs to finance his position. He will do so by borrowing from the treasury at the rate $\rfm$. Similarly, if the hedger{'}s position is positive, he will lend the cash amount to the treasury at the rate $\rfp$.

\paragraph{Collateral process and collateral account.}
The role of the collateral is to mitigate counterparty exposure of the two parties, i.e the potential loss on the transacted claim incurred by one party if the other defaults. The collateral process $C:={(C_t; \; t\geq 0)}$ is an $\mathbb{F}$ adapted process. We use the following sign conventions. If $C_t > 0$, the hedger is said to be the \textit{collateral provider}. In this case the counterparty measures a positive exposure to the hedger, hence asking him to post collateral so as to absorb potential losses arising if the hedger defaults. Vice versa, if $C_t < 0$, the hedger is said to be the \textit{collateral taker}, i.e., he measures a positive exposure to the counterparty and hence asks her to post collateral.

Collateral is posted and received in the form of cash in line with data reported by \cite{ISDA14}, according to which
cash collateral is the most popular form of collateral.\footnote{According to \cite{ISDA14} (see Table 3 therein), cash represents slightly more than $78\%$ of the total collateral delivered and these figures are broadly consistent across years. Government securities instead only constitute $18\%$ of total collateral delivered and other forms of collateral consisting of riskier assets, such as municipal bonds, corporate bonds, equity or commodities only represent a fraction slightly higher than $3\%$.}

We denote by $\rcp$ the rate on the collateral amount received by the hedger if he has posted the collateral, i.e., if he is the collateral provider, while $\rcm$ is the rate paid by the hedger if he has received the collateral, i.e., if he is the collateral taker. The rates $r_c^\pm$ typically correspond to Fed Funds or EONIA rates, i.e., to the contractual rates earned by cash collateral in the US and EURO markets, respectively. We denote by $B^{r_c^{\pm}}$ the cash accounts corresponding to these collateral rates, whose dynamics are given by
  \[
    dB_t^{r_c^{\pm}} = r_c^{\pm} B_t^{r_c^{\pm}} dt.
  \]
Moreover, let us define
  \[
    B_t^{r_c} := B_t^{r_c}(C) = e^{\int_0^t r_c(C_s) ds},
  \]
where
  \[
    r_c(x) = \rcp \ind_{\{x>0\}} + \rcm \ind_{\{x<0\}}.
  \]

Let $\psi_t^c$ be the number of shares of the collateral account $B^{r_c}_t$ held by the trader at time $t$. Then it must hold that
  \begin{equation}\label{eq:collrel}
    \psi_t^{c}B_t^{r_c}  = - C_t.
  \end{equation}
The latter relation means that if the trader is the collateral taker at $t$, i.e., $C_t < 0$, then he has purchased shares of the collateral account i.e., $\psi_t^{c} > 0$. Vice versa, if the trader is the collateral provider at time $t$, i.e., $C_t > 0$, then he has sold
shares of the collateral account to her counterparty.

Before proceeding further, we visualize in Figure \ref{fig:transflow} the mechanics governing the entire flow of transactions taking place.

\begin{figure}[ht]
    \centering
    \begin{tikzpicture}[thick,scale=0.9, every node/.style={transform shape}]
        \node[punkt, inner sep=10pt] (trader) {Trader};
        \node[punkt, inner sep=10pt,  left=2cm of trader] (funder) {Treasury Desk}
            edge[pil, bend left=35, blue, dotted] (trader)
            edge[pil, <-, bend right=35, blue, dotted] (trader)
            edge[pil, <-, bend left=10, blue, dashed] (trader)
            edge[pil, bend right=10, blue, dashed] (trader);
        \node[above left =1cm of trader] (rfp) {$\rfp$};
        \node[below left =1cm of trader] (rfm) {$\rfm$};
        \node[left =0.5cm of trader] (fundingcash) {Cash};
        \node[punkt, inner sep=10pt,  above=4cm of trader] (stock) {Stock \& Repo Market}
            edge[pil, <-, bend right=40, blue, dotted] (trader)
            edge[pil, bend left=40, blue, dotted] (trader)
            edge[pil, <-, bend left=20, red, solid] (trader)
            edge[pil, bend right=20, red, solid] (trader);
        \node[above =2cm of trader] (stocklending) {Stock};
        \node[below = 2cm of stock.west] (rrm) {$\rrm$};
        \node[below =2cm of stock.east] (rrp) {$\rrp$};
        \node[punkt, inner sep=10pt,  right=3.2cm of trader] (bond) {Bond Market}
            edge[pil, bend left=30, black, solid] (trader)
            edge[pil, <-, bend right=30, black, solid] (trader);
        \node[right =0.3cm of trader] (fundingcash) {Bonds $P^I$, $P^C$};
        \node[punkt, inner sep=10pt,  below=4cm of trader] (counter) {Counterparty}
            edge[pil, <-, bend right=50, blue, dotted] (trader)
            edge[pil, bend left=50, blue, dotted] (trader)
            edge[pil, <-, bend left=35, blue, dashed] (trader)
            edge[pil, bend right=35, blue, dashed] (trader);
        \node[below =1.5cm of trader] (collateralization) {Collateral};
        \node[below = 3.5cm of trader.west] (rcp) {$\rcp$};
        \node[below =3.5cm of trader.east] (rcm) {$\rcm$};
    \end{tikzpicture}
    \caption{Trading: Solid lines are purchases/sales, dashed lines borrowing/lending, dotted lines interest due; blue lines are cash, red lines are stock purchases for cash and black lines are bond purchases for cash.}
    \label{fig:transflow}
\end{figure}
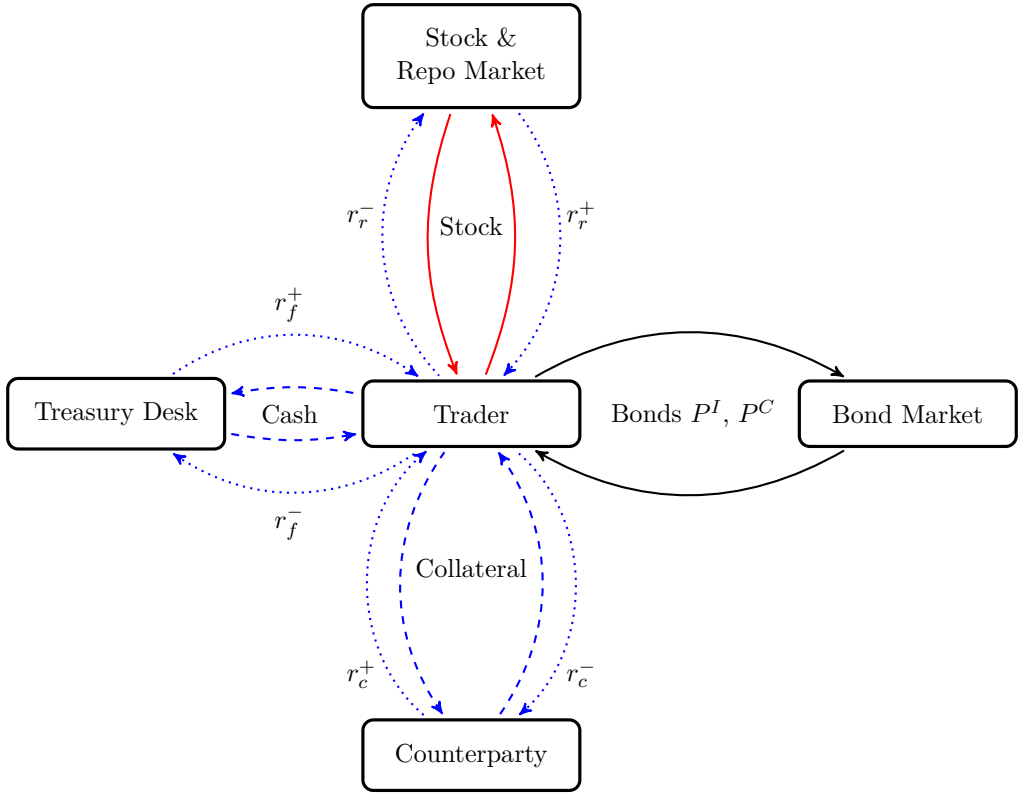

\section{Replicated claim, close-out value and wealth process} \label{sec:claim}
 We take the viewpoint of a trader who wants to replicate a European type claim on the stock security. Such a claim is purchased or sold by the trader from/to his counterparty over-the-counter and hence subject to counterparty credit risk. The closeout value of the claim is decided by a valuation agent who might either be one of the parties or a third party, in accordance with market practices as reviewed by the International Swaps and Derivatives Association {(ISDA)}. {The valuation agent determines the closeout value} of the transaction by calculating the Black Scholes price of the derivative using the discount rate $r_D$. Such a (publicly available) discount rate enables the hedger to introduce a valuation measure $\Qxx$ defined by the property that all securities have instantaneous growth rate $r_D$ under this measure. The rest of the section is organized as follows. We give the details of the valuation measure in Section \ref{sec:valuation}, give the price process of the claim to be replicated and the collateral process in Section \ref{sec:repclaim}, and define the closeout procedure in Section \ref{sec:closeout}. We define the class of admissible strategies in Section \ref{sec:repl}.

\subsection{The valuation measure} \label{sec:valuation}

We first introduce the default intensity model. Under the physical measure $\Px$, default times of trader and counterparty are assumed to be independent exponentially distributed random variables with constant intensities $h_i^{\Px}$, $i \in \{I,C\}$. It then holds that for each $i \in \{I, C\}$
  \[
    \varpi_t^{i,\Px} := H_t^i - \int_0^t\bigl(1-H_u^i\bigr)h_i^{\Px} \, du
  \]
is a $(\mathbb{G},\Px)$-martingale. The valuation measure $\Qxx$ chosen by the third party is equivalent to $\Px$ and is given by the Radon-Nikod\'{y}m density
  \begin{equation}\label{eq:q-girsanov}
    \frac{d\Qxx}{d\Px} \bigg|_{\mathcal{G}_{\tau}} = e^{\frac{r_D-\mu}{\sigma}W_{\tau}^{\Px} - \frac{(r_D-\mu)^2}{2\sigma^2}\tau} \Bigl(1+\frac{r^I - r_D}{h_I^{\Px}}\Bigr)^{H^I_\tau} e^{(r_D-r^I)\tau}\Bigl(1+\frac{r^C - r_D}{h_C^{\Px}}\Bigr)^{H^C_\tau} e^{(r_D-r^C)\tau},
  \end{equation}
where $r_D$ is the discount rate that the third party uses for valuation. We also recall that $r^I$ and $r^C$ denote the rate of returns of the bonds underwritten by the trader and counterparty respectively. Under $\Qxx$, the dynamics of the risky assets are given by
  \begin{align*}
    dS_t &= r_D S_t \,dt + \sigma S_t \,dW_t^{\Qxx}, \\
    dP_t^I & = r_D P_t^I \, dt - P_{t-}^I d\varpi_t^{I,\Qxx},\\
    dP_t^C & = r_D P_t^C \, dt - P_{t-}^C d\varpi_t^{C,\Qxx}
  \end{align*}
where $W^{\Qxx}:=(W^{\Qxx}_t; \; 0 \leq t \leq \tau)$ is a ${\Qxx}$-Brownian motion, while $\varpi^{I,\Qxx} := (\varpi_t^{I,\Qxx}; \; 0 \leq t \leq \tau)$ and $\varpi^{C,\Qxx}:=(\varpi_t^{C,\Qxx}; \;0 \leq t \leq \tau)$ are $\Qxx$-martingales. The above dynamics of $P_t^I$ and $P_t^C$ under the valuation measure $\Qxx$ can be deduced from their respective price processes given in \eqref{eq:priceproc} via a straightforward application of It\^{o}{'}s formula.

By application of Girsanov{'}s theorem, we have the following relations: $W^{\Qxx}_t = W^{\Px}_t + \frac{\mu-r_D}{\sigma} t$,
$\varpi_t^{i,\Qxx} = \varpi_t^{i,\Px} + \int_0^t \bigl(1-H_u^i\bigr) (h_i^{\Px} - h_i^{\Qxx}) du$. Here, for $i \in \{I,C\}$,
$h_i^{\Qxx} = r^i-r_D + h_i^{\Px}$ is the default intensity of the name $i$ under the valuation measure, which is assumed to be positive.

\subsection{Replicated claim and collateral specification} \label{sec:repclaim}
The price process of the claim ${\vartheta}$ to be replicated is, according to the third party{'}s valuation, given by
  \[
    \hat{V}(t,S_t) = e^{-r_D(T-t)} \mathbb{E}^{\Qxx}\bigl[ \Phi(S_T) \, \bigr\vert \, \mathcal{F}_t \bigr],
  \]
where $\Phi:\mathds{R}_{>0} \rightarrow \mathds{R}$ is a real valued function representing the terminal payoff of the claim. We will require that $\Phi$ is piecewise continuously differentiable and of at most polynomial growth. Additionally, the hedger has to post collateral for the claim. As opposed to the collateral used in the repo agreement, which is always the stock, the collateral mitigating counterparty credit risk of the claim is always cash. The collateral is chosen to be a fraction of the current exposure process of one party to the other. In case when the hedger sells a European call or put option on the security to his counterparty (he would then need to replicate the payoff $\Phi(S_T)$ which he needs to give to the counterparty at $T$), then the counterparty always measures a positive exposure to the hedger, while the hedger has zero exposure to the counterparty. As a result, the trader will always be the collateral provider, while the counterparty the collateral taker. By a symmetric reasoning, if the hedger buys a European call or put option from his counterparty (he would then replicate the payoff $-\Phi(S_T)$ received at maturity from his counterparty), then he will always be the collateral taker. On the event that neither the trader nor the counterparty have defaulted by time $t$, the collateral process is defined by
  \begin{equation}\label{eq:rulecoll}
    C_t = \alpha \hat{V}(t,S_t),
  \end{equation}
where $0 \leq \alpha \leq 1$ is the collateralization level. The case when $\alpha = 0$ corresponds to zero collateralization, $\alpha = 1$ gives full collateralization. Collateralization levels are industry specific and are reported on a quarterly basis by ISDA, see for instance \cite{ISDA11} , Table 3.3. therein.\footnote{The average collateralization level in 2010 across all OTC derivatives was $73.1\%$. Positions with banks and broker dealers are the most highly collateralized among the different counterparty types with levels around $88.6\%$. Exposures to non-financial corporations and sovereign governments and supra-national institutions tend to have the lowest collateralization levels, amounting to $13.9\%$.}

\subsection{Close-out value of transaction}\label{sec:closeout}

The ISDA market review of OTC derivative collateralization practices, see \cite{ISDA2010}, section 2.1.5, states that the surviving party should evaluate the transactions just terminated due to the default event, and claim for a reimbursement only after mitigating losses with the available collateral. In our study, we follow the risk-free closeout convention meaning that the trader liquidates his position at the counterparty default time at the market value. Next, we describe how we model it in our framework. Denote by $\theta_{\tau}$ the price of the hedging portfolio at $\tau$, where we recall that $\tau$ has been defined in Section \ref{sec:model}. The value of the portfolio depends on the value of the collateral account and on the residual value of the claim being traded at default. This term originates the well known credit and debit valuation adjustments terms. For a real number $x$, we use the notations $x^+ := \max(x,0)$, and $x^- := \max(0,-x)$. Then, we have
the following expression
  \begin{align}\label{eq:theta}
    \nonumber \theta_{\tau}(\hat{V}) & \phantom{:}=  \theta_{\tau}(C,\hat{V}) \\
    \nonumber & := \hat{V}(\tau,S_{\tau}) + \ind_{\{\tau_C<\tau_I \}} L_C \bigl(\hat{V}(\tau,S_{\tau})-C_{\tau-}\bigr)^- - \ind_{\{\tau_I<\tau_C \}} L_I \bigl(\hat{V}(\tau,S_{\tau}) - C_{\tau-}\bigr)^+\\
    \nonumber & \phantom{:}=  \min\Bigl((1-L_I)\bigl(\hat{V}(\tau,S_{\tau}) - C_{\tau-} \bigr) + C_{\tau-} \, , \, \hat{V}(\tau,S_{\tau})\Bigr) \ind_{\{\tau_I<\tau_C \}}\\
     &   \phantom{:=} + \max\Bigl((1-L_C)\bigl(\hat{V}(\tau,S_{\tau}) - C_{\tau-} \bigr) + C_{\tau-}\, , \, \hat{V}(\tau,S_{\tau})\Bigr) \ind_{\{\tau_C<\tau_I\}}\\
     \nonumber & \phantom{:}= \bigl(1-(1- \alpha) L_I\bigr)\hat{V}(\tau,S_{\tau})\ind_{\{\tau_I<\tau_C; \hat{V}(\tau_I, S_{\tau_I}) \geq 0\}} + \hat{V}(\tau,S_{\tau}) \ind_{\{\tau_I<\tau_C; \hat{V}(\tau_I, S_{\tau_I}) < 0\}}\\
    \nonumber & \phantom{:=} + \bigl(1-(1- \alpha) L_C\bigr)\hat{V}(\tau,S_{\tau})\ind_{\{\tau_C<\tau_I; \hat{V}(\tau_C, S_{\tau_C}) < 0\}\}} + \hat{V}(\tau,S_{\tau}) \ind_{\{\tau_C<\tau_I; \hat{V}(\tau_C, S_{\tau_C}) \geq 0\}}
  \end{align}
where $\ind_{\{\tau_C<\tau_I \}} L_C \bigl(\hat{V}(\tau,S_{\tau}) - C_{\tau-}\bigr)^-$ originates the residual $\CVA$ type term after collateral mitigation, while $\ind_{\{\tau_I<\tau_C\}}  L_I \bigl(\hat{V}(\tau,S_{\tau}) - C_{\tau-}\bigr)^+$ originates the $\DVA$ type term. Here $0 \leq L_I \leq 1$ and $0 \leq L_C \leq 1$ are the loss rates against the trader and counterparty claims, respectively.

\begin{remark}
We present two specific cases to better explain the mechanics of the close-out. Recall that we have assumed that the collateralization level $0 \leq \alpha \leq 1$.
  \begin{itemize}
    \item The trader sold a call option to the counterparty (hence $\hat{V}(t,S_{t})>0$ for all $t$). This means that in each time $t$ the trader is the collateral provider, $C_t = \alpha \hat{V}(t,S_{t}) > 0$, given that the counterparty always measures a positive exposure to the trader.
      \begin{itemize}
	\item The counterparty defaults first and before the maturity of the claim. Then the trader will net the amount $\hat{V}(\tau,S_{\tau})$ owed to the counterparty with his collateral provided to the counterparty, and only return to her the residual amount $\hat{V}(\tau,S_{\tau}) - C_{\tau-}$. As a result, we expect his trading strategy (exclusive of collateral) to replicate this amount. We next verify that this is the case. Using the above expression of closeout wealth (which is inclusive of collateral) we obtain ${\theta_{\tau}(\hat{V})} = \hat{V}(\tau,S_{\tau})$. Hence, his closeout amount exclusive of collateral is $\hat{V}(\tau,S_{\tau}) - C_{\tau-}$ coinciding with the amount which he must give to the counterparty.
	\item The trader defaults first and before the maturity of the claim. Then the counterparty will only get a recovery fraction $(1-L_I) \bigl(\hat{V}(\tau,S_{\tau})-C_{\tau-}\bigr)$ from the trader. As a result, we expect the strategy of the trader (exclusive of collateral) to replicate this amount. Next, we verify that this is the case. Using the above expression of closeout wealth (inclusive of collateral) we obtain ${\theta_{\tau}(\hat{V})} = \hat{V}(\tau,S_{\tau})-L_I  (\hat{V}(\tau,S_{\tau}) - C_{\tau-})$. Hence, his closeout amount exclusive of collateral is $(1-L_I)\bigl(\hat{V}(\tau,S_{\tau}) - C_{\tau-}\bigr)$ coinciding with the amount which must be returned to the counterparty.
      \end{itemize}
    \item The trader purchased a call option from the counterparty (hence $\hat{V}(t,S_{t}) <0$ for all $t$). This means that in each time $t$ the trader is the collateral taker, $C_t = \alpha \hat{V}(t,S_{t}) < 0$, given that the trader always measures a positive exposure to the counterparty.
      \begin{itemize}
	\item The trader defaults first and before the maturity of the claim. Then the counterparty will net the amount $-\hat{V}(\tau,S_{\tau})$ which she should return to the trader with the collateral held by the trader, and only return to him the residual amount $-(\hat{V}(\tau,S_{\tau}) - C_{\tau-})$. This means that the wealth process of the trader exclusive of collateral at $\tau_I$ should be equal to $\hat{V}(\tau,S_{\tau}) - C_{\tau-})$. We next verify that this is the case. Using the above expression of closeout wealth (which is inclusive of collateral) we obtain ${\theta_{\tau}(\hat{V})} = \hat{V}(\tau,S_{\tau})$. Hence, the closeout amount exclusive of collateral is $\hat{V}(\tau,S_{\tau}) - C_{\tau-}$ coinciding with the above wealth process.
	\item The counterparty defaults first and before the maturity of the claim. Then the trader will only get a recovery fraction $-(1-L_C) \bigl(\hat{V}(\tau,S_{\tau})-C_{\tau-}\bigr)$ from the counterparty.
        As a result, we expect the wealth process of the trader (exclusive of collateral) at $\tau_I$ to be equal to $(1-L_C) \bigl(\hat{V}(\tau,S_{\tau}) - C_{\tau-} \bigr)$. Next, we verify that this is the case. Using the above expression of closeout wealth (inclusive of collateral) we obtain ${\theta_{\tau}(\hat{V})} = \hat{V}(\tau,S_{\tau})-L_C ((\hat{V}(\tau,S_{\tau}) - C_{\tau-}))$. Hence, the closeout payout exclusive of collateral is $(1-L_C)\bigl(\hat{V}(\tau,S_{\tau}) - C_{\tau-}\bigr)$ coinciding with the above mentioned wealth process.
      \end{itemize}
  \end{itemize}
\end{remark}

\subsection{The wealth process} \label{sec:repl}

We allow for collateral to be fully rehypothecated. This means that the collateral taker is granted an unrestricted right to use the collateral amount, i.e. he can use it to purchase investment securities. This is in agreement with most ISDA annexes, including the New York Annex, English Annex, and Japanese Annex. We notice that for the case of cash collateral, the percentage of re-hypotheticated collateral amounts to about $90\%$ (see Table 8 in \cite{ISDA14}) hence largely supporting our assumption of full collateral re-hypothecation. As in \cite{br}, the collateral received can be seen as an ordinary component of a hedger{'}s trading strategy, although this applies only prior to the counterparty{'}s default. We denote by $V_t({\bm\varphi})$ the \textit{legal} wealth process of the hedger, and use $V_t^C({\bm\varphi})$ to denote the \textit{actual} wealth process. In other words, we have $V_t^C({\bm\varphi}) = V_t({\bm\varphi}) - C_t$, where we recall that $C_t < 0$ indicates that the hedger is the collateral taker. This reflects the fact that the collateral remains legally in the hands of the provider whereas it can actually be used by the taker (via rehypothecation) and invested in the risky assets.

Let ${\bm\varphi} := \bigl(\xi_t,\xi_t^f,\xi_t^I, \xi_t^C \; t \geq 0\bigr)$. Here, we recall that $\xi$ denotes the number of shares of the security, and $\xi^{f}$ the number of shares in the funding account. Moreover, $\xi^I$ and $\xi^C$ are the number of shares of trader and counterparty risky bonds, respectively. Recalling Eq.~\eqref{eq:collrel}, and expressing all positions in terms of number of shares multiplied by the price of the corresponding security, the wealth process $V({\bm \varphi})$ is given by the following expression
  \begin{equation}\label{eq:wealth}
    V_t({\bm\varphi}) := \xi_t S_t + \xi_t^I P_t^I + \xi_t^C P_t^C + \xi_t^f B_t^{r_f} + \psi_t B_t^{r_r} - \psi_t^{c} B_t^{r_c},
  \end{equation}
where we notice that the number of shares $\psi$ of the repo account and the number of shares $\psi^{c}$ held in the collateral account are uniquely determined by equations~\eqref{eq:selff} and~\eqref{eq:collrel}, respectively.

\begin{definition}
A collateralized trading strategy ${\bm\varphi}$ is \textit{self-financing} if, for $t \in [0,T]$, it holds that
  \[
    V_t({\bm\varphi}) := V_0({\bm\varphi}) + \int_0^t \xi_u \, dS_u + \int_0^t \xi_u^I \, dP_u^I + \int_0^t \xi_u^C \, dP_u^C + \int_0^t \xi_u^f \, dB_u^{r_f} +  \int_0^t \psi_u \, dB_u^{r_r} -\int_0^t \psi_u^{c} \, dB_u^{r_c},
  \]
where $V_0({\bm\varphi})$ is an arbitrary real number.
\end{definition}

Moreover, we define the class of admissible strategies as follows:

\begin{definition}\label{def:control-U}
The admissible control set is a class of $\mathbb{F}$-predictable locally bounded trading strategies, such that the portfolio process is bounded from below, see also \cite{Delbaen}.
\end{definition}

\section{Arbitrage-free pricing and XVA} \label{sec:BSDEform}

The goal of this section is to find prices for the derivative security with payoff $\Phi(S_T)$ that are free from arbitrage in a certain sense. Before discussing arbitrage-free prices, we have to make sure that the underlying market does not admit arbitrage from the hedger's perspective (as discussed in \cite[Section 3]{br}). In the underlying market, the trader is only allowed to borrow/lend stock, buy/sell risky bonds and borrow/lend from the funding desk. In particular, neither the derivative security, nor the collateral process is involved.

\begin{definition}
 We say that the market $(S_t,P_t^I,P_t^C)$ admits a \textit{hedger's arbitrage} if we can find controls $(\xi^f,\xi,\xi^I,\xi^C)$ such that for some initial capital $x \geq 0$, denoting the wealth process with this initial capital as $V_t(x)$, $t \geq 0$, we have that $\Px \bigl[V_\tau(x) \geq e^{\rfp \tau}V_0(x) \bigr] =1$ and $\Px\bigl[V_\tau(x) > e^{\rfp \tau}V_0(x)\bigr] >0$. If the market does not admit hedger's arbitrage for all $x \geq 0$, we say that the market is arbitrage free from the hedger's perspective.
\end{definition}

In the sequel, we make the following standing assumption:

\begin{assumption}\label{ass:necessary}
The following relations hold between the different rates: $r_r^{+} \le r_f^{-}$, $r_f^+ \leq r_f^-$, $r_f^+ \vee r_D < r^I + h_I^\Px$ and $r_f^+ \vee r_D < r^C + h_C^\Px$.
\end{assumption}

\begin{remark}\label{rem:measure}
The above assumption is necessary to preclude arbitrage. The condition $r_D < \left(r^I + h_I^\Px\right) \wedge \left(r^C + h_C^\Px\right)$ is needed for the existence of the valuation measure as discussed at the end of section \ref{sec:valuation} ($h_i^{\Qxx} = r^i + h_i^\Px-r_D$, $i=I,C$, and risk-neutral default intensities must be positive). If, by contradiction, $r_r^{+} > r_f^{-}$, the trader can borrow cash from the funding desk at the rate $r_f^{-}$ and lend it to the repo market at the rate $r_r^{+}$, while holding the stock as a collateral. This results in a sure win for the trader. Similarly, if the trader could fund his strategy from the treasury at a rate $r_f^- < r_f^+$, it would clearly result in an arbitrage. The condition $r_f^+ < r^I + h_I^\Px$ (and mutatis mutandis $r_f^+ < r^C + h_C^\Px$) has a more practical interpretation: as
  \[
    dP^I_t = r^I  P^I_t \, dt - P^I_{t-} \, d\varpi_t^{I,\Px}  = (r^I + h_I^{\Px})  P^I_t \, dt - P^I_{t-} \, dH_t^{I,\Px},
  \]
it precludes the arbitrage opportunity of short selling the bond underwritten by the trader{'s} firm with value $P_t^I$ and investing the proceeds in the funding account.

\end{remark}

We next provide a sufficient condition guaranteeing that the underlying market is free of arbitrage.

\begin{proposition}\label{thm:arb-market}
Assume, in addition to Assumption \ref{ass:necessary}, that $r_r^+ \leq r_f^+ \leq r_r^-$. Then the model does not admit arbitrage opportunities for the hedger for any $x \geq 0$.
\end{proposition}
We remark that in a market model without defaultable securities, similar inequalities between borrowing and lending rates have been derived by \cite{br} (Proposition 3.3), and by \cite{NiRut} (Proposition 3.1). We impose additional relations between lending rates and return rates of the risky bonds given that our model also allows for counterparty risk.

\begin{proof}
First, observe that under the conditions given above we have
  \begin{align*}
    \nonumber   r_r \psi_t& = r_r^+ \psi_t \ind_{\{\psi_t>0\}} + r_r^- \psi_t \ind_{\{\psi_t< 0\}} \leq r_f^+ \psi_t \ind_{\{\psi_t>0\}} + r_f^+ \psi_t \ind_{\{\psi_t < 0\}} = r_f^+ \psi_t \\
    r_f \xi_t^f & = r_f^+ \xi_t^f \ind_{\{\xi_t^f >0\}} + r_f^- \xi_t^f \ind_{\{{\xi_t^f} < 0\}} \leq r_f^+ \xi_t^f \ind_{\{\xi_t^f >0\}} + r_f^+ \xi_t^f\ind_{\{\xi_t^f < 0\}} = r_f^+ \xi_t^f
  \end{align*}
Next, it is convenient to write the wealth process under a suitable measure $\tilde{\Px}$ specified via the stochastic exponential
  \[
    \frac{d\tilde{\Px}}{d\Px} \bigg|_{\mathcal{G}_{\tau}} = e^{\frac{r_f^+-\mu}{\sigma}W_{\tau}^{\Px} - \frac{(r_f^+-\mu)^2}{2\sigma^2}\tau} \Bigl(1+\frac{r^I - r_f^+}{h_I^{\Px}}\Bigr)^{H^I_\tau} e^{(r_f^+-r^I)\tau}\Bigl(1+\frac{r^C - r_f^+}{h_C^{\Px}}\Bigr)^{H^C_\tau} e^{(r_f^+-r^C)\tau}
  \]
By Girsanov's theorem, $\tilde{\Px}$ is an equivalent measure to $\Px$ such that the dynamics of the risky assets are given by
  \begin{align*}
    dS_t & = r_f^+ S_t \, dt + \sigma S_t dW_t^{\tilde{\Px}},\\
    dP_t^I & = r_f^+ P_t^I \, dt - P_{t-}^I d\varpi_t^{I,\tilde{\Px}},\\
    dP_t^C & = r_f^+ P_t^C \, dt - P_{t-}^C d\varpi_t^{C,\tilde{\Px}}
  \end{align*}
where $W^{\tilde{\Px}}:=(W^{\tilde{\Px}}_t; \; 0 \leq t \leq \tau)$ is a $\tilde{\Px}$ Brownian motion, while $\varpi^{I,\tilde{\Px}} := (\varpi_t^{I,\tilde{\Px}}; \; 0 \leq t \leq \tau)$ and $\varpi^{C,\tilde{\Px}}:=(\varpi_t^{C,\tilde{\Px}}; \;0 \leq t \leq \tau)$ are $\tilde{\Px}$-martingales. The $r_f^+$ discounted assets $\tilde{S}_t := e^{-r_f^+ t} S_t$, $\tilde{P}_t^I := e^{-r_f^+ t} P_t^I$ and $\tilde{P}_t^C := e^{-r_f^+ t} P_t^C$ are thus $\tilde{\Px}$-martingales. In particular, $W^{\tilde{\Px}} = W^{\Px} + \frac{\mu-r_f^+}{\sigma}$ and for $i \in \{I,C\}$, trader and counterparty default intensity under $\tilde{\Px}$ are given by $h_i^{\tilde{\Px}} = r^i-r_f^+ + h_i^{\Px} > 0$ in light of the assumptions of the proposition.

Denote the wealth process associated with $(S_t,P_t^I,P_t^C)_{t \geq 0}$ in the underlying market by $V^0_t$. Using the self-financing condition, its dynamics are given by
  \begin{align*}
    \nonumber dV^0_t &= \bigl(r_f \xi_t^f B_t^{r_f} + r_f^+ \xi_t S_t + r_r \psi_tB_t^{r_r} + r_f^+ \xi_t^I P_t^I + r_f^+ \xi_t^C P_t^C \bigr) \, dt  \\
    \nonumber & \phantom{=}+ \xi_t \sigma S_t \, dW_t^{\tilde{\Px}}  - \xi_{t-}^I P_{t-}^I  \, d\varpi_t^{I,\tilde{\Px}}  - \xi_{t-}^C P_{t-}^C \, d\varpi_t^{C,\tilde{\Px}}\\
    & = \bigl(r_f \xi_t^f B_t^{r_f} + r_r \psi_t B_t^{r_r}  \bigr) \, dt + \xi_t \,d S_t +\xi_{t-}^I \, dP_t^I + \xi_{t-}^C \, dP_t^C.
  \end{align*}
Assume that the initial capital at disposal of the trader is $x \geq 0$. Then we have that
  \begin{align*}
    V^0_\tau(x) -V^0_0(x) & = \int_0^\tau \bigl( r_f^+ \xi_t S_t + r_f \xi_t^f B_t^{r_f} + r_r \psi_tB_t^{r_r} + r_f^+\xi_t^I P^I_t + r_f^+\xi_t^C P^C_t  \bigr) \, dt  \\
    \nonumber & \phantom{=}+ \int_0^\tau \xi_t \sigma S_t \, dW_t^{\tilde{\Px}} - \int_0^\tau\xi_{t-}^I P_{t-}^I \, d\varpi_t^{I, \tilde{\Px}} - \int_0^\tau\xi_{t-}^C P_{t-}^C\, d \varpi_t^{C,\tilde{\Px}} \\
    & \leq \int_0^\tau \bigl( r_f^+ \xi_t S_t + r_f^+ \xi_t^f B_t^{r_f} + r_f^+ \psi_tB_t^{r_r} + r_f^+\xi_t^I P^I_t + r_f^+\xi_t^C P^C_t \bigr) \, dt  \\
    \nonumber & \phantom{=}+ \int_0^\tau \xi_t \sigma S_t \, dW_t^{\tilde{\Px}} - \int_0^\tau\xi_{t-}^I P_{t-}^I \, d\varpi_t^{I, \tilde{\Px}} - \int_0^\tau\xi_{t-}^C P_{t-}^C\, d \varpi_t^{C,\tilde{\Px}} \\
    & = \int_0^\tau r_f^+ V^0_t(x) \, dt + \int_0^\tau \xi_t \sigma S_t \, dW_t^{\tilde{\Px}} - \int_0^\tau\xi_{t-}^I P_{t-}^I \, d\varpi_t^{I, \tilde{\Px}} - \int_0^\tau\xi_{t-}^C P_{t-}^C\, d \varpi_t^{C,\tilde{\Px}}.
  \end{align*}
Therefore, it follows that
  \[
    e^{-r_f^+\tau} V^0_{\tau}(x) -V^0_0(x) \leq \int_0^\tau \xi_t \, d\tilde{S}_t - \int_0^{\tau}\xi_{t-}^I \, d\tilde{P}_{t-}^I - \int_0^{\tau}\xi_{t-}^C \, d\tilde{P}_{t-}^C.
  \]
Note now that the right hand side is a local martingale bounded from below (as the value process is bounded from below by the admissibility condition), and therefore is a supermartingale. Taking expectations, we conclude that
  \[
    \Exx^{\tilde{\Px}} \bigl[  e^{-r_f^+\tau} V^0_\tau(x) -V^0_0(x) \bigr] \leq 0
  \]
and therefore either $\tilde{\Px} \bigl[  V^0_\tau(x) = e^{r_f^+\tau} x \bigr] = 1$ or $\tilde{\Px} \bigl[  V^0_\tau(x) < e^{r_f^+\tau} x \bigr] >0 $. As $\tilde{\Px}$ is equivalent to $\Px$, this shows that arbitrage opportunities for the hedger are precluded in this model (as $r_f^+$ is the riskless rate he would receive by lending the positive cash amount $x$ to the funding desk).
\end{proof}

Next we want to define the notion of an arbitrage free price of a derivative security from the hedger{'}s perspective.

\begin{definition}\label{def:arb-price}
A valuation $P \in \mathbb{R}$ of a derivative security with terminal payoff $\vartheta \in \mathcal{F}_T$ is called \textit{hedger's arbitrage-free} if for all $\gamma \in \mathbb{R}$, buying $\gamma$ securities for the price $\gamma P$ and hedging in the market with an admissible strategy and zero initial capital, does not create hedger's arbitrage.
\end{definition}

Before giving a characterization of hedger's arbitrage prices, we first need to return to the wealth dynamics.  Since our goal is to define the total valuation adjustment under the valuation measure $\Qxx$, we first rewrite the dynamics of the wealth process under this measure. Using the condition~\eqref{eq:selff}, we obtain from \eqref{eq:wealth} that
  \begin{align}\label{eq:vtlast}
    \nonumber dV_t &= \bigl(r_f \xi_t^f B_t^{r_f} + (r_D - r_r) \xi_t S_t + r_D \xi_t^I P_t^I + r_D \xi_t^C P_t^C - r_c \psi_t^c B_t^{r_c} \bigr) \, dt  \\
    \nonumber & \phantom{=}+ \xi_t \sigma S_t \, dW_t^{\Qxx}  - \xi_{t-}^I P_{t-}^I \, d\varpi_t^{I,\Qxx}  - \xi_{t-}^C P_{t-}^C  \, d\varpi_t^{C,\Qxx}\\
    \nonumber &= \Bigl( \rfp \bigl(\xi_t^f B_t^{r_f}\bigr)^+ -\rfm \bigl(\xi_t^f B_t^{r_f}\bigr)^- +  (r_D - \rrm) \bigl(\xi_t S_t\bigr)^+ - (r_D - \rrp) \bigl(\xi_t S_t\bigr)^- + r_D \xi_t^I P_t^I + r_D \xi_t^C P_t^C\Bigr) \, dt \\
    \nonumber & \phantom{=}\Bigl(- r_c^- \bigl( \psi_t^c B_t^{r_c}\bigr)^+ +  r_c^+ \bigl(\psi_t^c B_t^{r_c}\bigr)^-  \Bigr) \, dt  \\
    \nonumber & \phantom{=}+ \xi_t \sigma S_t \, dW_t^{\Qxx} - \xi_{t-}^I P_{t-}^I \, d\varpi_t^{I,\Qxx}  - \xi_{t-}^C P_{t-}^C  \, d\varpi_t^{C,\Qxx} \\
    \nonumber &= \Bigl( \rfp \bigl(\xi_t^f B_t^{r_f}\bigr)^+ -\rfm \bigl(\xi_t^f B_t^{r_f}\bigr)^- +  (r_D - \rrm) \bigl(\xi_t S_t\bigr)^+ - (r_D - \rrp) \bigl(\xi_t S_t\bigr)^- + r_D \xi_t^I P_t^I + r_D \xi_t^C P_t^C\Bigr) \, dt \\
    \nonumber & \phantom{=}- \Bigl(r_c^+ \bigl(C_t\bigr)^+ - r_c^- \bigl(C_t\bigr)^-  \Bigr) \, dt  \\
    & \phantom{=}+ \xi_t \sigma S_t \, dW_t^{\Qxx}  - \xi_{t-}^I P_{t-}^I   \, d\varpi_t^{I,\Qxx}  - \xi_{t-}^C P_{t-}^C d\varpi_t^{C,\Qxx}
  \end{align}
Set now
  \begin{align}\label{eq:Zetas}
    \nonumber Z_t &= \xi_t \sigma S_t\\
    \nonumber Z^I_t &= -\xi_{t-}^I P_{t-}^I\\
    Z_t^C &= -\xi_{t-}^C P_{t-}^C.
  \end{align}
Moreover, using again the condition~\eqref{eq:selff} with \eqref{eq:wealth}, we obtain
  \begin{align*}
    \xi_t^f B_t^{r_f} &= V_t(\bm \varphi) - \xi_t^I P^I_t - \xi_t^C P^C_t + \psi_t^{c} B_t^{r_c} \\
    &= V_t(\bm \varphi) - \xi_t^I P^I_t - \xi_t^C P^C_t - C_t
  \end{align*}
Then the dynamics~\eqref{eq:vtlast} reads as
  \begin{align}\label{eq:vtlast2}
    \nonumber dV_t &= \Bigl(\rfp \bigl(V_t + Z_t^I + Z_t^C - C_t\bigr)^+ -\rfm \bigl(V_t + Z_t^I + Z_t^C - C_t\bigr)^-\\
    \nonumber & \phantom{=} +  (r_D - \rrm) \frac{1}{\sigma}\bigl(Z_t\bigr)^+ -  (r_D - \rrp) \frac{1}{\sigma}\bigl(Z_t\bigr)^- - r_D Z_t^I - r_D Z_t^C  - \Bigl(r_c^+ \bigl(C_t\bigr)^+ - r_c^- \bigl(C_t\bigr)^-  \Bigr) \,  \Bigr) dt \\
    & \phantom{=} + Z_t\, dW_t^{\Qxx} + Z_t^I \, d\varpi_t^{I,{ \Qxx}}  + Z_t^C \, d\varpi_t^{C,\Qxx}
  \end{align}

We next define
  \begin{align}
    f^+\bigl(t,v,z,z^I,z^C; \hat{V}\bigr) &= -\Bigl(\rfp \bigl(v+ z^I + z^C- \alpha \hat{V}_t\bigr)^+ -\rfm \bigl(v + z^I + z^C - \alpha \hat{V}_t\bigr)^- \nonumber\\ & \phantom{=:}+  (r_D - \rrm) \frac{1}{\sigma}z^+  -  (r_D - \rrp) \frac{1}{\sigma}z^- - r_D z^I - r_D z^C \nonumber\\
    & \phantom{=:} + r_c^+ \bigl(\alpha \hat{V}_t\bigr)^+ - r_c^-\bigl(\alpha \hat{V}_t\bigr)^-  \Bigr)\\
    f^-\bigl(t,v,z,z^I,z^C; \hat{V}\bigr) &= - f^+\bigl(t,-v,-z,-z^I,-z^C; -\hat{V}_t\bigr)
  \end{align}
where the driver depends on the market valuation process $(\hat{V}_t)$ (via the collateral $(C_t)$). In particular $f^\pm \, : \, \Omega \times [0,T] \times \R^4$, $(\omega, t,v,z,z^I,z^C) \mapsto f^\pm\bigl(t,v,z,z^I,z^C; \hat{V}_t(\omega)\bigr)$ are drivers of BSDEs as discussed in Appendix \ref{App_BSDE}. Moreover, define $V^+(\gamma)$, $V^-(\gamma)$ as solutions of the BSDEs
  \begin{equation}\label{eq:BSDE-sell}
    \left\{ \begin{array}{rl} -dV_t^+(\gamma) &= f^+\bigl(t,V_t^+,Z_t^+,Z_t^{I,+},Z_t^{C,+}; \hat{V}\bigr) \, dt - Z^+_t\, dW_t^{\Qxx} - Z_t^{I,+} \, d\varpi_t^{I,\Qxx}  - Z_t^{C,+} \, d\varpi_t^{C,\Qxx}\\
    V_\tau^+(\gamma) & = \gamma \Bigl(\theta_\tau(\hat{V}) \ind_{\{\tau<T\}} + \Phi(S_T)\ind_{\{\tau = T\}}\Bigr) \end{array}\right.
  \end{equation}
and
  \begin{equation}\label{eq:BSDE-buy}
    \left\{ \begin{array}{rl} -dV_t^-(\gamma) &= f^-\bigl(t,V_t^-,Z_t^-,Z_t^{I,-},Z_t^{C,-}; \hat{V}\bigr) \, dt - Z^-_t\, dW_t^{\Qxx} - Z_t^{I,-} \, d\varpi_t^{I,\Qxx}  - Z_t^{C,-} \, d\varpi_t^{C,\Qxx}\\
    V_\tau^-(\gamma) & = \gamma \Bigl(\theta_\tau(\hat{V}) \ind_{\{\tau<T\}} + \Phi(S_T)\ind_{\{\tau = T\}}\Bigr) \end{array}\right.
  \end{equation}
We note that $\bigl(V_t^+(\gamma)\bigr)$ describes the wealth process when replicating the claim $\gamma \Phi(S_T)$ for $\gamma >0$ (hence hedging the position after selling $\gamma$ securities with terminal payoff $\Phi(S_T)$). On the other hand, $\bigl(-V_t^-(\gamma)\bigr)$ describes the wealth process when replicating the claim $-\gamma \Phi(S_T)$, $\gamma>0$ (hence hedging the position after buying $\gamma$ securities with terminal payoff $\Phi(S_T)$). To ease the notation we set $V_t^+(1) = V_t^+$ and $V_t^-(1) = V_t^-$. We note also that the two BSDEs are intrinsically related: $(V_t^-,Z_t^-,Z_t^{I,-},Z_t^{C,-})$ is a solution to the data $\bigl(f^-, \theta_\tau(\hat{V}), \Phi(S_T)\bigr)$ exactly if $(-V_t^-,-Z_t^-,-Z_t^{I,-},-Z_t^{C,-})$ is a solution to the data $\bigl(f^+, \theta_\tau(-\hat{V}), -\Phi(S_T)\bigr)$.

\begin{theorem}\label{thm:arb-price}
Let $\Phi$ be a function of polynomial growth. Assume that
  \begin{equation}\label{eq:A14}
    r_r^+ \leq r_f^+ \leq r_r^-, \qquad r_f^+ \leq r_f^-, \qquad r_f^+ \vee r_D < r^I + h_I^{\Px}, \qquad r_f^+ \vee r_D < r^C + h_C^{\Px},
  \end{equation}
and
  \begin{equation}\label{eq:comp}
 \rcp \vee \rcm \le \rfm \leq \bigl(r^I + h_I^{\Px}\bigr) \wedge \bigl( r^C + h_C^{\Px}\bigr).
  \end{equation}
If $V_0^- \leq V_0^+$, (where $(V_t^+)$ and $(V_t^-)$ are the (first components of) solutions of the BSDEs \eqref{eq:BSDE-sell} and \eqref{eq:BSDE-buy}), then there exist prices $\pi^{sup}$ and $\pi^{inf}$, $\pi^{inf} \leq \pi^{sup}$, (called hedger's upper and lower arbitrage price) for the option $\Phi(S_T)$ such that all prices in the closed interval $[\pi^{inf}, \pi^{sup}]$ are free of hedger's arbitrage. All prices strictly bigger than $\pi^{sup}$ and strictly smaller than $\pi^{inf}$ provide then arbitrage opportunities for the hedger. In particular, we have that $\pi^{sup} = V_0^+$ and $\pi^{inf} = V_0^-$.
\end{theorem}

\begin{proof}
First, notice that by virtue of the conditions in \eqref{eq:A14}, the underlying market model is free of hedger's arbitrage. Notice that by positive homogeneity of the drivers $f^+$ and $f^-$ of the BSDEs \eqref{eq:BSDE-sell} and \eqref{eq:BSDE-buy}, it is enough to consider the cases with $\gamma =1$. We note that the polynomial growth of $\Phi$ entails that $\vartheta =\Phi(S_T) \in  L^2(\Omega, \mathcal{F}_T, \Qxx)$. As \eqref{eq:BSDE-sell} and \eqref{eq:BSDE-buy} satisfy the conditions (A1) -- (A3) of Assumption \ref{BSDE_ass}, the BSDEs have unique solutions.

Next, we note that we can perfectly hedge the option with terminal payoff $\Phi(S_T)$ using the initial capital $V_0^+$. Thus it is clear that any price $P > V_0^+$ is not arbitrage free, as we could just sell the option for that price, use $V_0^+$ to hedge the claim and put $P-V_0^+$ in the funding account. Using the same argument it is true that when buying an option, any price $P<V_0^-$ will lead to arbitrage.

Thus if $V_0^- \leq V_0^+$ we can conclude that all arbitrage free prices are within the interval $[\pi^{inf} = V_0^- , V_0^+ = \pi^{sup}]$, whereas if $V_0^- > V_0^+$ there are no arbitrage free prices.

Second, assume by contradiction that a valuation $P  \leq V_0^{+}$ would lead to an arbitrage when selling the option. This means that starting with initial capital $P$, the trader can perfectly hedge a claim with terminal payoff $\vartheta' \in \mathcal{F}_T$, where $\vartheta' \geq \vartheta =\Phi(S_T)$ a.s. and $\Px[\vartheta'>\vartheta]>0$. We have

\begin{align*}
&f^+(t,\bar{u}, \bar{z},\bar{z}^I ,\bar{z}^C; \hat{V}_t^{\vartheta'}) - f^+(t,\bar{u}, \bar{z},\bar{z}^I ,\bar{z}^C; \hat{V}_t^{\vartheta}) \\
&=  \bigl(\rfm-\rfp\bigr)\left(\bigl(\bar{u}+ \bar z^I + \bar z^C +(1- \alpha) \hat{V}_t^{\vartheta'}\bigr)^+   - \bigl(\bar{u}+ \bar z^I + \bar z^C+(1- \alpha) \hat{V}_t^{\vartheta}\bigr)^+  \right) + \alpha \rfm \left(  \hat{V}_t^{\vartheta'} -  \hat{V}_t^{\vartheta} \right) \\
&+ \alpha \rcm \Bigl((\hat{V}_t^{\vartheta'})^- - (\hat{V}_t^{\vartheta})^-\Bigr) - \alpha \rcp \Bigl(\bigl(\hat{V}_t^{\vartheta'}\bigr)^+ - \bigl(\hat{V}_t^{\vartheta}\bigr)^+\Bigr)\\
&\ge  \alpha \bigl({\rfm}- (\rcp  \vee \rcm)  \bigr)\bigl(\hat{V}_t^{\vartheta'} - \hat{V}_t^{\vartheta}\bigr) \geq 0,
 \end{align*}
where for the first inequality we have used that fact that $\alpha <1$, and that $\rfm\ge \rfp,$ and    \eqref{eq:comp} for the second inequality.

Notice that $\vartheta' \geq \vartheta$ implies that $\hat{V}_t^{\vartheta'} \geq \hat{V}_t^{\vartheta}$, which in turn leads to the following inequality between the closeout terms: $\theta_\tau(\hat{V}_t^{\vartheta'})  \ind_{\{\tau<T\}} \geq \theta_\tau(\hat{V}_t^{\vartheta})  \ind_{\{\tau<T\}}$ (see also \eqref{eq:theta} for the definition of the closeout).  Thus, we can apply the comparison principle Theorem \ref{thm:compare}. It thus follows that $V_t^{\vartheta'} \geq V_t^{\vartheta}$ and in particular $P > V_0^+$ (using strict comparison, i.e., $P = V_0$ implies $\vartheta' =\vartheta$ a.s.), contradicting our assumption. Using a symmetric argument, it follows that $P  \geq V_0^{-}$.	Thus, if $V_0^{-} \leq V_0^{+}$, we can conclude that all valuations in the interval $[\pi^{inf} = V_0^{-} , V_0^{+} = \pi^{sup}]$ are arbitrage-free, whereas no arbitrage free valuation exists if $V_0^{-} > V_0^{+}$.
\end{proof}

Our goal is to compute the \textit{total valuation adjustment} $\XVA$ that has to be added to the Black-Scholes price to get the actual price. As we have seen, the situation is asymmetric for sell- and buy-prices, so we will have to define different sell- and buy $\XVA$s.
\begin{definition}
The seller's $\XVA$ is the $\mathbb{G}$-adapted stochastic process $(\XVA_t^{sell})$ defined as
  \[
    \XVA_t^{sell} := V^+_t - \hat{V}(t,S_t)
  \]
while the buyer's $\XVA$ is defined as
  \[
    \XVA_t^{buy} := V^-_t - \hat{V}(t,S_t).
  \]
\end{definition}

$XVA^{sell}$ corresponds to the total costs (including collateral servicing costs, funding costs, and counterparty adjustment costs) that the hedger incurs when replicating the payoff of an option he sold, whereas $\XVA^{buy}$ corresponds to the total costs (including collateral servicing costs, funding costs, and counterparty adjustment costs) that the hedger incurs when replicating the payoff of an option he purchased. As we will see in Section \ref{sec:pitdef}, these two $\XVA$s agree only in the case where the drivers of the BSDEs are linear. We note that the difference of the $\XVA$s also describes the width of the no-arbitrage interval, as
  \[
    \XVA_0^{sell} - \XVA_0^{buy} = V^+_0 - V^-_0 = \pi^{sup} - \pi^{inf}.
  \]

\section{Explicit examples}\label{sec:expexp}

We specialize our framework to deal with a concrete example for which we can provide fully explicit expressions of the total valuation adjustment. More specifically, we consider an extension of \cite{Piterbarg}'s model accounting for counterparty credit risk and closeout costs. This means that defaultable bonds of trader and counterparty become an integral part of the hedging strategy. Throughout the section, we make the following assumptions on the interest rates as in Piterbarg{'}s setup:
\begin{align}\label{eq:Piterrates}
r_f^+ &= r_f^- = r_f \nonumber \\
r_c^+ &= r_c^- = r_c \nonumber \\
r_D &= r_r^+ = r_r^- = r_r
\end{align}
The difference between the discount rate $r_D$ chosen by the third party and the repo rate may also be interpreted as a proxy for illiquidity of the repo market. Under this interpretation, $r_D=r_r$ would indicate a regime of full liquidity. \cite{Piterbarg} mentions that the case to be expected in practice is $r_f>r_r>r_c$ (and $r_f \geq r_r$ is indeed necessary for no-arbitrage reasons, see Remark \ref{rem:measure}).

\subsection{Piterbarg's model} \label{sec:pitdef}

As in \cite{Piterbarg} we do not take into account default risk of the trader and counterparty, hence implying that the transaction is always terminated at $\tau = T$. Under this setting, using the collateral specification from Section \ref{sec:repclaim}, we obtain an explicit, closed-form expression of the total value adjustment as well as of the optimal (super)-hedging strategy.

We first note that equations~\eqref{eq:BSDE-sell} and~\eqref{eq:BSDE-buy} coincide by linearity and a transformation $\tilde{V}_t = \bigl(B_t^{r_f} \bigr)^{-1} V_t^+ = \bigl(B_t^{r_f}\bigr) ^{-1} V_t^-$ simplifies the situation. Writing the BSDE in integral form we have
  \begin{align}\label{PitebargBSDE}
    \nonumber \tilde{V}_t &= \bigl(B_T^{r_f}\bigr)^{-1} \Phi(S_T)  - \int_t^T \tilde{Z}_s dW_s^{\Qxx} + \int_t^T \bigl(B_s^{r_f}\bigr)^{-1} \bigl({r_f - r_c}\bigr) \alpha \hat{V}(s,S_s) ds - \int_t^T \bigl({r_D - r_r}\bigr) \frac{1}{\sigma} \tilde{Z}_s ds \\
    &= \bigl(B_T^{r_f}\bigr)^{-1} \Phi(S_T)  - \int_t^T \tilde{Z}_s dW_s^{\Qxx} + \int_t^T \bigl(B_s^{r_f}\bigr)^{-1} \bigl({r_f - r_c}\bigr) \alpha \hat{V}(s,S_s) ds,
  \end{align}
where we have used the collateral specification given in \eqref{eq:rulecoll}, and the assumption that $r_D = r_r$. This immediately implies that in case when all rates $r_f,\, r_r$, $r_c$ are the same, then the total valuation adjustment becomes zero. Otherwise, it is nonzero and given by two main contributions:
\begin{enumerate}
\item The first term captures the adjustment due to discounting the terminal value of the claim at the funding rate $r_f$. Since $r_f$ is usually higher than the risk-free rate in the Black-Scholes model, we would obtain that $\mathbb{E}^{\Qxx}\bigl[\bigl(B_T^{r_f}\bigr)^{-1} \Phi(S_T) \bigr]$ is smaller than the corresponding quantity evaluated using the risk-free rate.
\item The last term captures the adjustment, proportional to the size of posted collateral $\alpha \hat{V}(t,S_t)$, originating from the difference between funding and collateral rates.
\end{enumerate}
Using the Clark-Ocone formula, we can find $\tilde{Z}_t$ by means of Malliavin Calculus. We have that
  \[
    \tilde{Z}_t = \mathbb{E}^{\Qxx} \biggl[\Bigl. D_t \Bigl( \bigl(B_T^{r_f}\bigr)^{-1} \Phi(S_T) +  \int_t^T e^{-r_f s }\bigl(r_f - r_c\bigr) \alpha \hat{V}(s,S_s) \,ds \Bigr) \Bigr\vert \mathcal{F}_t\biggr],
  \]
where $D_t$ denotes the Malliavin derivative which may be computed as
  \begin{align}\label{eq:malliavin}
    & \phantom{==} D_t \Bigl( \bigl(B_T^{r_f}\bigr)^{-1} \Phi(S_T) + \int_t^T \bigl(B_s^{r_f}\bigr)^{-1} \bigl(r_f - r_c\bigr) \alpha \hat{V}(s,S_s) ds \Bigr)\nonumber\\
    &= \bigl(B_T^{r_f}\bigr)^{-1} \Phi'(S_T) D_t S_T + \alpha \int_t^T \bigl(B_s^{r_f}\bigr)^{-1} \bigl(r_f - r_c\bigr) D_t \hat{V}(s,S_s)\,ds \nonumber\\
    & = \bigl(B_T^{r_f}\bigr)^{-1} \Phi'(S_T) \sigma S_T + \alpha \int_t^T \bigl(B_s^{r_f}\bigr)^{-1} \bigl(r_f - r_c\bigr) \frac{\partial}{\partial x}\hat{V}(s,S_s)\sigma S_T \, ds.
  \end{align}
Above, we have used the chain rule of Malliavin calculus and the well known fact that $D_t S_T = \sigma S_T$.\footnote{We refer to \cite{Nualartbook} for an introduction to Malliavin derivatives.} We expect that the $\tilde{Z}$ term of the BSDE would correspond to an ``adjusted delta hedging'' strategy, with the delta hedging strategy recovered if all rates are identical.
Indeed, let us denote the market Delta by
  \[
    \hat{\Delta}(t,S_t) := \frac{\partial}{\partial x} \hat{V}(t, S_t) = \frac{\partial}{\partial x} \mathbb{E}^{\Qxx}\Bigl[\frac{B_t^{r_r}}{B_T^{r_r}}  \Phi(S_T) \vert \mathcal{F}_t\Bigr]
  \]
and note that by a martingale argument
  \[
    \hat{\Delta}(t,S_t) \frac{S_t}{B_t^{r_r}}  = \mathbb{E}^{\Qxx}\Bigl[ \Bigl. \Phi'(S_T) \frac{S_T}{B_T^{r_r}} \Bigr\vert \mathcal{F}_t\Bigr] = \mathbb{E}^{\Qxx}\Bigl[ \Bigl. \frac{S_T}{B_T^{r_r}} \hat{\Delta}(T,S_T) \Bigr\vert \mathcal{F}_t\Bigr].
  \]
Then the Malliavin derivative from \eqref{eq:malliavin} can be written in terms of Delta as
  \[
    D_t\bigl(...\bigr) =  \frac{1}{B_T^{r_f}} \Phi'(S_T) \sigma S_T + \alpha \int_t^T \frac{1}{B_s^{r_f}} (r_f - r_c) \hat{\Delta}(s,S_s)\sigma S_s \, ds.
  \]
Therefore we get
  \begin{align}
    \nonumber \tilde{Z}_t & = \sigma \mathbb{E}^{\Qxx}\Bigl[\Bigl.\frac{1}{B_T^{r_f}} \Phi'(S_T)  S_T\Bigr\vert \mathcal{F}_t\Bigr] + \alpha \int_t^T \frac{1}{B_s^{r_f}} (r_f - r_c) \sigma \mathbb{E}^{\Qxx}\bigl[\bigl.\hat{\Delta}(s,S_s) S_s\bigr\vert \mathcal{F}_t\bigr] \, ds \\
    \nonumber & =  \sigma \frac{B_T^{r_r}}{B_T^{r_f}} \mathbb{E}^{\Qxx}\Bigl[\Bigl.\Phi'(S_T) \frac{S_T}{B_T^{r_r}}  \Bigr\vert \mathcal{F}_t\Bigr] + \alpha \int_t^T \frac{1}{B_s^{r_f}} (r_f - r_c) \sigma B_s^{r_r} \mathbb{E}^{\Qxx}\Bigl[\Bigl.\hat{\Delta}(s,S_s) \frac{S_s}{B_s^{r_r}}\Bigr\vert \mathcal{F}_t\Bigr] \, ds\\
    \nonumber & =  \sigma \frac{B_T^{r_r}}{B_T^{r_f}} \frac{S_t}{B_t^{r_r}} \hat{\Delta}(t,S_t) + \alpha \int_t^T \frac{1}{B_s^{r_f}} (r_f - r_c) \sigma B_s^{r_r} \mathbb{E}^{\Qxx}\Bigl[\Bigl.\hat{\Delta}(s,S_s) \frac{S_s}{B_s^{r_r}}\Bigr\vert \mathcal{F}_t\Bigr] \, ds\\
    \nonumber & =  \frac{1}{B_T^{r_f}} \frac{B_T^{r_r}}{B_t^{r_r}} \sigma S_t \hat{\Delta}(t,S_t) + \alpha \int_t^T \frac{1}{B_s^{r_f}} (r_f - r_c) \sigma B_s^{r_r} \hat{\Delta}(t,S_t) \frac{S_t}{B_t^{r_r}} \, ds\\
    & =  \frac{1}{B_T^{r_f}} \frac{B_T^{r_r}}{B_t^{r_r}} \sigma S_t \hat{\Delta}(t,S_t) + \alpha \Bigl(1-\frac{r_c}{r_f}\Bigr) \biggl(\frac{1}{ B_t^{r_f}} - \frac{1}{B_T^{r_f}}\biggr) \frac{B_T^{r_r}}{B_t^{r_r}} \sigma S_t  \hat{\Delta}(t,S_t) \\
    & =  \frac{1}{B_T^{r_f}} \frac{B_T^{r_r}}{B_t^{r_r}} \sigma S_t \hat{\Delta}(t,S_t) + \alpha \frac{1}{B_t^{r_r}} \Bigl(\frac{r_f-r_c}{r_r-r_f}\Bigr) \biggl(\frac{B_T^{r_r}}{ B_T^{r_f}} - \frac{B_t^{r_r}}{B_t^{r_f}}\biggr) \sigma S_t  \hat{\Delta}(t,S_t) \label{eq:finalz}
  \end{align}
We clearly see that $\tilde{Z}_t$ is square integrable and therefore the stochastic integral in \eqref{PitebargBSDE} is a true martingale. Thus we can calculate $\tilde{V}_t$ directly by taking conditional expectations:
  \begin{align}
    \nonumber \tilde{V}_t &= \mathbb{E}^{\Qxx}\Bigl[\bigl.\bigl(B_T^{r_f}\bigr)^{-1} \Phi(S_T) \bigr\vert \mathcal{F}_t \Bigr] + \alpha (r_f - r_c) \int_t^T
    \bigl(B_s^{r_f} \bigr)^{-1} B_s^{r_r} \mathbb{E}^{\Qxx}\Bigl[\bigl. \bigl(B_s^{r_r}\bigr)^{-1} \hat{V}({s},S_{s}) \bigr\vert \mathcal{F}_t \Bigr] d{ s} \\
    &= \bigl(B_T^{r_f}\bigr)^{-1} B_T^{r_r} \frac{1}{B_t^{r_r}} \hat{V}(t,S_t) + \alpha (r_f - r_c) \int_t^T
    \bigl(B_s^{r_f} \bigr)^{-1} B_s^{r_r} \bigl(B_t^{r_r}\bigr)^{-1} \hat{V}(t,S_t) ds
    \label{eq:BSDE1}
  \end{align}

In the second step, we have used the fact that $(B_t^{r_r})^{-1} \hat{V}(t,S_t)$ is a $\mathcal{F}_t$ martingale.
Hence, we obtain
  \begin{align}
    \nonumber \tilde{V}_t &= \frac{B_T^{r_r}}{B_T^{r_f}} \frac{1}{B_t^{r_r}}\hat{V}(t,S_t) + \alpha (r_f - r_c) \frac{1}{B_t^{r_r}} \hat{V}(t,S_t) \int_t^T \frac{ B_s^{r_r}}{B_s^{r_f}}  ds \\
    &=\frac{B_T^{r_r}}{B_T^{r_f}} \frac{1}{B_t^{r_r}}\hat{V}(t,S_t) + \alpha \frac{1}{B_t^{r_r}} \frac{r_f - r_c}{r_f - r_r}\Bigl(\frac{B_t^{r_r}}{B_t^{r_f}} - \frac{B_T^{r_r}}{B_T^{r_f}}\Bigr) \hat{V}(t,S_t) \label{eq:splitpricePiterbarg}\\
    &=\frac{1}{B_t^{r_r}} \biggl(\frac{B_T^{r_r}}{B_T^{r_f}} + \alpha \frac{r_f - r_c}{r_f - r_r}\Bigl(\frac{B_t^{r_r}}{B_t^{r_f}} - \frac{B_T^{r_r}}{B_T^{r_f}}\Bigr)\biggr) \hat{V}(t,S_t) \nonumber \\
    &:=\beta_t \hat{V}(t,S_t) \label{eq:unitedpricePiterbarg}
  \end{align}
The above equation \eqref{eq:unitedpricePiterbarg} gives the price in terms of a unique adjustment factor, while \eqref{eq:splitpricePiterbarg} tracks separately the adjustment that comes from using a particular funding rate and that one coming from the collateralization procedure. We can therefore directly write down the total valuation adjustment as
  \begin{equation}\label{eq:FVAPiterbarg}
    \XVA_t = V_t - \hat{V}_t  = B_t^{r_f} \tilde{V}_t - \hat{V}_t = \Biggl({\frac{B_t^{r_f}}{B_t^{r_r}}} \biggl( \frac{B_T^{r_r}}{B_T^{r_f}} + \alpha \frac{r_f - r_c}{r_f - r_r}\Bigl(\frac{B_t^{r_r}}{B_t^{r_f}} - \frac{B_T^{r_r}}{B_T^{r_f}}\Bigr)\biggr) - 1 \Biggr) \hat{V}(t,S_t)
  \end{equation}
From the representation~\eqref{eq:FVAPiterbarg}, we can see that it is obtained as a percentage of the publicly available price $\hat{V}(t,S_t)$ of the claim. Note that from \eqref{eq:Zetas}, $\tilde{Z}_t = \bigl(B_t^{r_f}\bigr)^{-1} \xi_t\sigma S_t$. Using this along with \eqref{eq:finalz} we obtain that the number of stock shares $\xi_t$ held by the hedger in his replication strategy is given by
  \begin{align}
    \xi_t & = \frac{B_t^{r_f}}{B_T^{r_f}} \frac{B_T^{r_r}}{B_t^{r_r}} \hat{\Delta}(t,S_t) + \frac{B_t^{r_f}}{B_t^{r_r} } \left(\alpha \frac{r_f - r_c}{r_f - r_r}\Bigl( \frac{B_t^{r_r}}{B_t^{r_f}} - \frac{B_T^{r_r}}{B_T^{r_f}}\Bigr) \right) \hat{\Delta}(t,S_t)\label{eq:splitDeltaPiterbarg}\\
    & = \beta_t B_t^{r_f} \hat{\Delta}(t,S_t). \label{eq:unitedDeltaPiterbarg}
  \end{align}
The representation \eqref{eq:splitDeltaPiterbarg} gives a decomposition of the strategy in a first part consisting of a Delta hedging strategy adjusted for funding costs, and a second part giving a correction due to the gap between funding and collateral rates. The representation  \eqref{eq:unitedDeltaPiterbarg} gives instead an overall adjustment factor. Let us finally note that the extension to only almost everywhere differentiable payoffs (as European put and call options) follows directly from the a.s. convergence of the Black-Scholes prices and Deltas.

We next analyze the dependence of $\XVA$ on funding rates and collateralization levels. Figure \ref{fig:nodef} shows that the XVA is negative when the collateralization level is small. This is consistent with the expression~\eqref{eq:FVAPiterbarg}, and can be understood as follows. In the classical Black-Scholes case, the hedger borrows and lends money at rate $r_D=r_r$. In our case, however, he borrows the cash needed to buy the stock at the repo rate $r_r$ (cash-driven transaction). He invests the remaining amount coming from the received option premium at the funding rate $r_f>r_r$. By contrast, in a Black-Scholes world the hedger would have invested this amount at the lower rate $r_r$. This explains why the XVA is negative. As $\alpha$ gets larger, the trader also needs to finance purchases of the collateral which he needs to post to the counterparty. In order to do so, he borrows from the treasury at the rate $r_f$. However, he would only receive interests at rate $r_c$ on the posted collateral. This yields a loss to the trader given that $r_c < r_r < r_f$. Figure  \ref{fig:nodef} confirms our intuition. It also shows that the stock position of the trader decreases as the funding rate $r_f$ increases, and increases if the collateralization level $\alpha$ increases.

  \begin{figure}
    \includegraphics[width=0.49\textwidth]{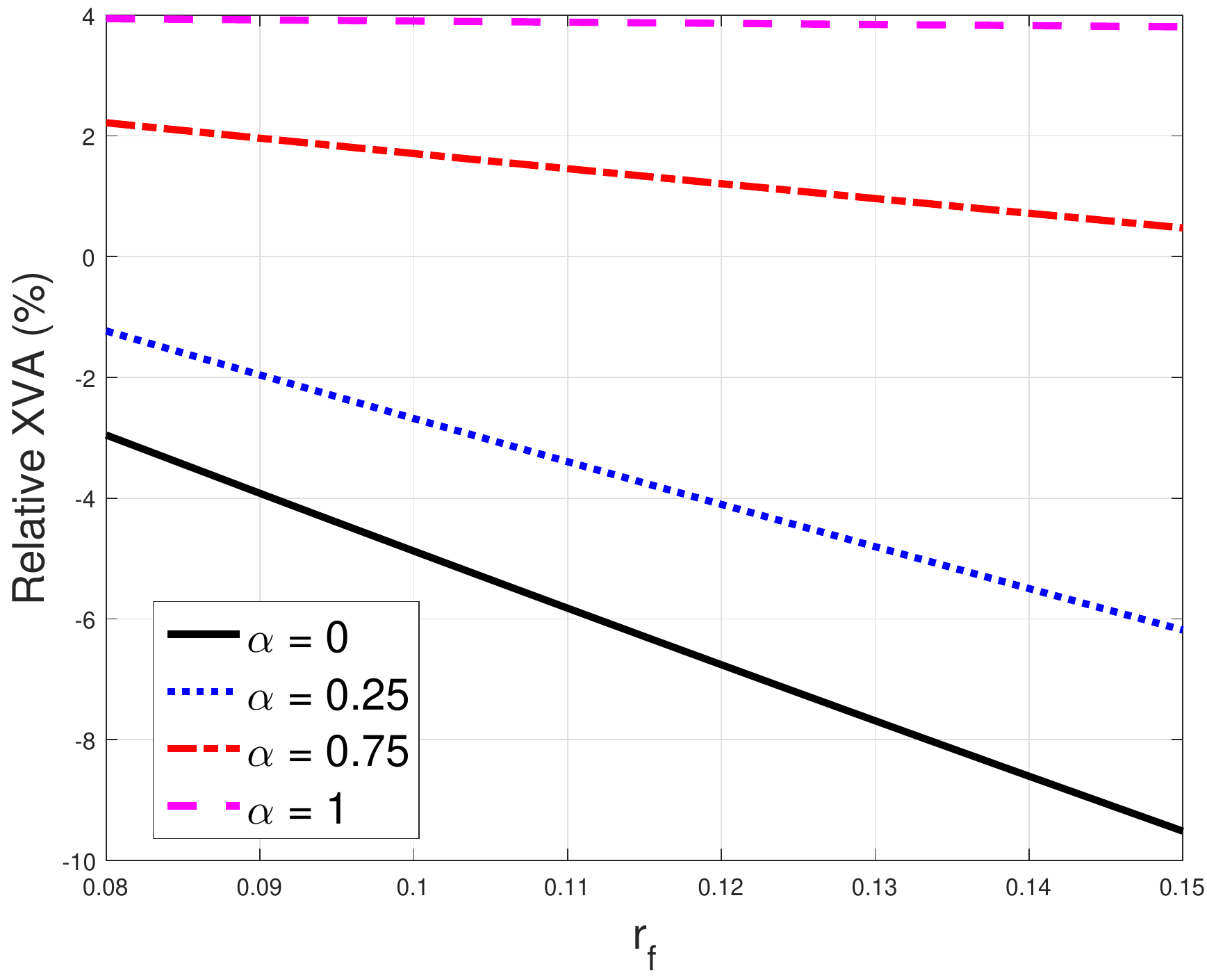}
    \includegraphics[width=0.49\textwidth]{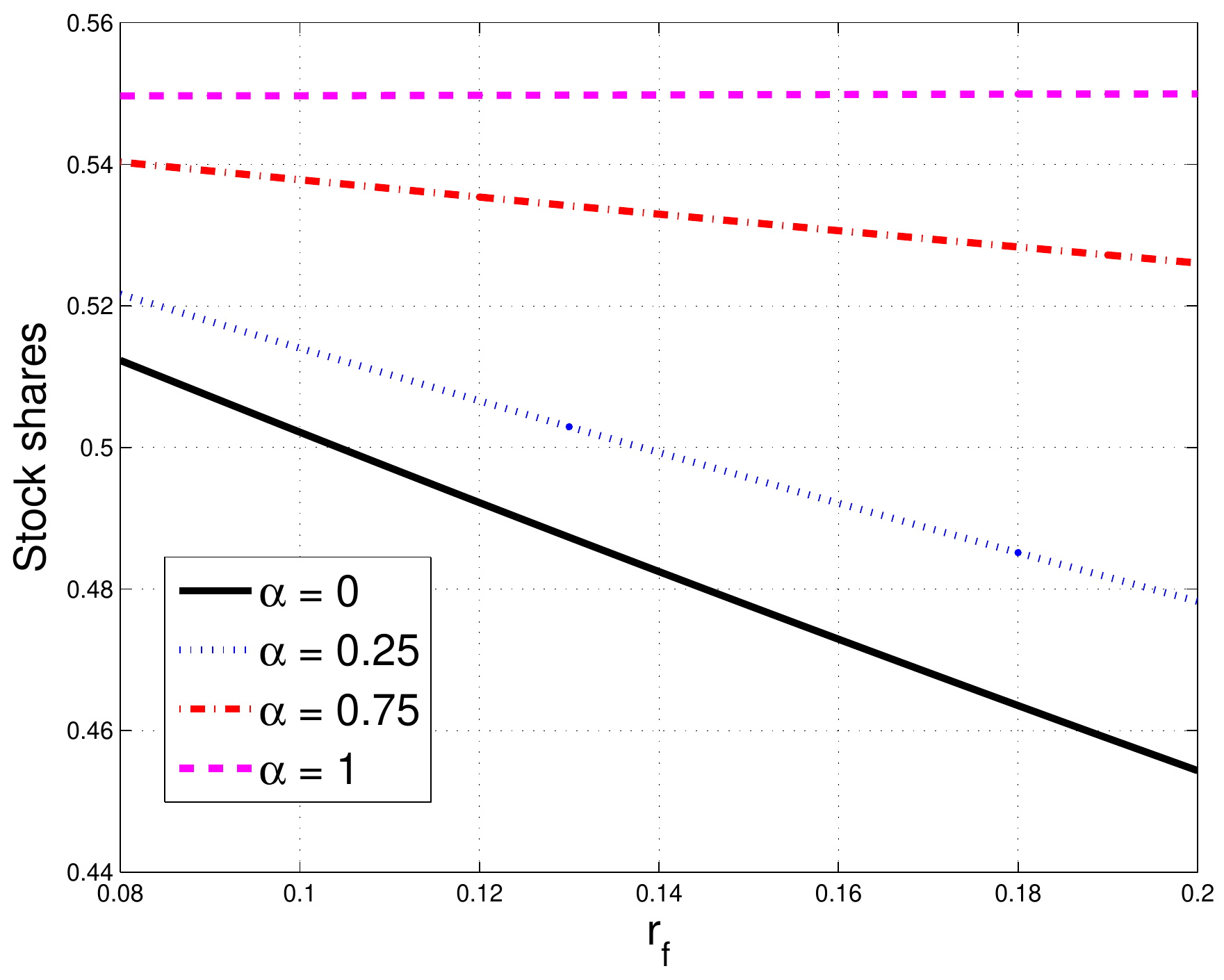}
    \caption{Left panel: XVA  as a function of $r_f$ for different collateralization levels $\alpha$. Right panel: Number of stock shares in the replication strategy. We set $r_D = 0.05$, $r_c = 0.01$, $\sigma = 0.2$, and $\alpha = 0$. The claim is an at-the-money European call option with maturity $T=1$.}
  \label{fig:nodef}
  \end{figure}

\subsection{Piterbarg's model with defaults}

In this section, we extend Piterbarg's model by including the possibility of a default by investor or counterparty. Thus we are working under the  general framework put forward in this paper. Again, we assume the rates equalities in~\eqref{eq:Piterrates} to be consistent with Piterbarg{'}s setup.

\subsubsection{The general representation formula}

As for the case of Piterbarg model without defaults treated in the previous section, the BSDEs \eqref{eq:BSDE-sell} and \eqref{eq:BSDE-buy} agree by linearity, and switching to $r_f$-discounted quantities they reduce to
 \begin{align}\label{eq:piterbargjumpbsde}
    \left\{
      \begin{array}{rl}
	- d\tilde{V}_t &= \Bigl(\bigl(r_D - r_f \bigl) \tilde{Z}_t^I + \bigl(r_D - r_f \bigl) \tilde{Z}_t^C \bigr)  + \bigl(r_f -r_c\bigr) \tilde{C}_t \Bigr) \, dt\\
	&\phantom{=} - \tilde{Z}_t\, dW_t^{\Qxx}  - \tilde{Z}_t^I \, d\varpi_t^{I,\Qxx}  - \tilde{Z}_t^C \, d\varpi_t^{C,\Qxx} \\
	\tilde{V}_\tau &= \bigl(B_T^{r_f}\bigr)^{-1} \Phi(S_T)\ind_{\{\tau = T\}} + \bigl(B_\tau^{r_f}\bigr)^{-1}\theta_{\tau}(\hat{V}) \ind_{\{ \tau < T\}}
      \end{array}
    \right.
  \end{align}

\begin{proposition}
The BSDE \eqref{eq:piterbargjumpbsde} admits an a.s. unique solution given by the explicit representation
  \begin{align*}
    \tilde{V}_t \ind_{\{\tau \geq t\}} & = \mathbb{E}^{\Qxx} \biggl[ \biggl.  \Bigl( \bigl(B_T^{r_f}\bigr)^{-1} \Phi(S_T) \ind_{\{\tau = T\}} + \bigl(B_\tau^{r_f}\bigr)^{-1} \theta_{\tau}(\hat{V})   \ind_{\{t < \tau <T\}} \Bigr) \Gamma_t^\tau \biggr. \biggr. \\
    & \phantom{=====} \biggl. \biggl. - \int_{t \wedge \tau}^\tau \bigl(r_c -r_f\bigr) \tilde{C}_s \Gamma_t^s \, ds \, \biggr\vert \, \mathcal{G}_t \biggr]
  \end{align*}
where
  \[
    \Gamma_t^s := \Bigl(1+ \frac{r_D-r_f}{\hIQ}\Bigr)^{H_{s \wedge \tau}^I- H_{t \wedge \tau}^I}e^{(r_f - r_D)(s\wedge \tau - t\wedge \tau )} \Bigl(1+ \frac{r_D-r_f}{\hCQ}\Bigr)^{H_{s\wedge \tau }^C- H_{t \wedge \tau }^C}e^{(r_f - r_D)(s\wedge \tau -t \wedge \tau)}
  \]
\end{proposition}

\begin{proof}
Existence and uniqueness follows from Theorem \ref{thm:main}. We follow the classical approach for linear BSDEs, e.g. as given for the L\'{e}vy-process case by \cite[Section 3]{Quenez} to construct the concrete solution. We note first that $(\Gamma_t^s)_{\{ s \ge t\} }$ is a Dol\'{e}ans-Dade exponential and satisfies
  \[
    \Gamma_t^s  = 1 + \int_{t \wedge \tau}^{s \wedge \tau} \Gamma_t^{u-} \frac{r_D-r_f}{\hIQ} d\varpi_u^{I,{\Qxx}} +  \int_{t \wedge \tau}^{s \wedge \tau} \Gamma_t^{u-} \frac{r_D-r_f}{\hCQ} d\varpi_u^{C,{\Qxx}}.
  \]
Therefore, we have
  \begin{align*}
    \tilde{V}_\tau \Gamma_t^\tau - \tilde{V}_{t \wedge \tau} \Gamma_t^{t \wedge \tau} & = \int_{t \wedge \tau}^\tau \tilde{V}_{s-} \, d \Gamma_t^s  + \int_{t \wedge \tau}^\tau \Gamma_t^{s-} \, d\tilde{V}_s + \bigl[\tilde{V}_\mathbf{\cdot}, \Gamma_t^\mathbf{\cdot}\bigr]_{\tau} - \bigl[\tilde{V}_\mathbf{\cdot}, \Gamma_t^\mathbf{\cdot}\bigr]_{t \wedge \tau}\\
    & = \int_{t \wedge \tau}^\tau  \tilde{V}_{s-} \Gamma_t^{s-} \frac{r_D-r_f}{\hIQ} \, d\varpi_s^{I,{\Qxx}} + \int_{t \wedge \tau}^\tau  \tilde{V}_{s-} \Gamma_t^{s-} \frac{r_D-r_f}{\hCQ} \, d\varpi_s^{C,{\Qxx}} + \int_{t \wedge \tau}^\tau  \Gamma_t^{s-}\tilde{Z}_s \, dW_s^{\Qxx} \\
    & \phantom{=} + \int_{t \wedge \tau}^\tau   \Gamma_t^{s-}\tilde{Z}_s^I \, d\varpi_s^{I,{\Qxx}}   + \int_{t \wedge \tau}^\tau   \Gamma_t^{s-}\tilde{Z}_s^C \, d\varpi_s^{C,{\Qxx}} \\
    & \phantom{=} - \int_{t \wedge \tau}^\tau  \Gamma_t^{s-} \Bigl(\bigl(r_D-r_f\bigr)\tilde{Z}_s^I + \bigl(r_D-r_f\bigr)\tilde{Z}_s^C  + \bigl(r_f - r_c\bigr)\tilde{C}_s \Bigr) \, ds \\
    & \phantom{=} + \int_{t \wedge \tau}^\tau   \Gamma_t^{s-} \frac{r_D-r_f}{\hIQ}\tilde{Z}_s^I \, d\varpi_s^{I,\Qxx}  + \int_{t \wedge \tau}^\tau   \Gamma_t^{s-} \frac{r_D-r_f}{\hCQ}\tilde{Z}_s^C \, d\varpi_s^{C,\Qxx} \\
    & \phantom{=} + \int_{t \wedge \tau}^\tau   \Gamma_t^{s-} \bigl(r_D-r_f\bigr)\tilde{Z}_s^I \, ds + \int_{t \wedge \tau}^\tau   \Gamma_t^{s-} \bigl(r_D-r_f\bigr)\tilde{Z}_s^C\, ds \\
    & = - \int_{t \wedge \tau}^\tau  \Gamma_t^{s-} \bigl(r_f - r_c\bigr)\tilde{C}_s \, ds + M_\tau - M_{t \wedge \tau}
  \end{align*}
where $(M_s)$ given by
  \begin{align*}
    M_s & =  \int_{t \wedge \tau}^{s \wedge \tau}  \Gamma_t^{u-} \Bigl(\bigl(\tilde{V}_{u-} + \tilde{Z}_u^I  \bigr)\frac{r_D-r_f}{\hIQ} + \tilde{Z}_u^I \Bigr) \, d\varpi_u^{I,{\Qxx}} \\
    & \phantom{=} + \int_{t \wedge \tau}^{s \wedge \tau} \Gamma_t^{u-} \Bigl(\bigl(\tilde{V}_{u-} + \tilde{Z}_u^C  \bigr)\frac{r_D-r_f}{\hCQ} + \tilde{Z}_u^C \Bigr) \, d\varpi_u^{C,{\Qxx}} + \int_{t \wedge \tau}^{s\wedge \tau}  \Gamma_t^{u-}\tilde{Z}_u \, dW_u^\Qxx
  \end{align*}
is a true martingale (as all coefficients are bounded). Multiplying with $\ind_{\{\tau \geq t\}}$ and taking conditional expectations with respect to $\mathcal{G}_t$ yields
  \[
    \tilde{V}_t \ind_{\{\tau \geq t\}} = \tilde{V}_t \Gamma_t^t \ind_{\{\tau \geq t\}} =  \mathbb{E}^{\Qxx} \biggl[ \biggl.  \tilde{V}_\tau \Gamma_t^\tau \ind_{\{\tau \geq t\}} + \int_{t \wedge \tau}^\tau \Gamma_t^{s-} \bigl(r_f - r_c\bigr)\tilde{C}_s \, ds \, \biggr\vert \, \mathcal{G}_t \biggr]
  \]
and thus the solution.
\end{proof}

Using Eq.~\eqref{eq:theta} along with the collateral specification given in Eq.~\eqref{eq:rulecoll}, we note in particular that we can also write the representation as
  \begin{align}
    \tilde{V}_t \ind_{\{\tau \geq t\}}  & = \mathbb{E}^{\Qxx} \Bigl[ \Bigl.  \bigl(B_T^{r_f}\bigr)^{-1} \Phi(S_T)  \Gamma_t^T  \ind_{\{\tau = T\}} \, \Bigr\vert \, \mathcal{G}_t \Bigr] \nonumber \\
    & \phantom{=}+ \mathbb{E}^{\Qxx} \Bigl[ \Bigl.   \bigl(B_{\tau_I}^{r_f}\bigr)^{-1}\bigl(1-(1- \alpha) L_I\bigr)\hat{V}(\tau_I,S_{\tau_I}) \Gamma_t^{\tau_I}  \ind_{\{t < \tau_I < \tau_C \wedge T; \hat{V}(\tau_I,S_{\tau_I}) \geq 0\}} \Bigr.\Bigr.\nonumber\\
    & \phantom{= + \mathbb{E}^{\Qxx} \Bigl[}  \Bigl.\Bigl. + \bigl(B_{\tau_I}^{r_f}\bigr)^{-1}\hat{V}(\tau_I,S_{\tau_I}) \Gamma_t^{\tau_I}  \ind_{\{t < \tau_I < \tau_C \wedge T; \hat{V}(\tau_I,S_{\tau_I}) < 0\}}  \, \Bigr\vert \, \mathcal{G}_t \Bigr] \nonumber \\
    & \phantom{=}+ \mathbb{E}^{\Qxx} \Bigl[ \Bigl.   \bigl(B_{\tau_C}^{r_f}\bigr)^{-1}\bigl(1-(1- \alpha) L_C\bigr)\hat{V}(\tau_C,S_{\tau_C}) \Gamma_t^{\tau_C}  \ind_{\{t < \tau_C < \tau_I \wedge T; \hat{V}(\tau_C,S_{\tau_C}) < 0\}} \Bigr.\Bigr.\nonumber \\
    & \phantom{= + \mathbb{E}^{\Qxx} \Bigl[}  \Bigl.\Bigl. + \bigl(B_{\tau_C}^{r_f}\bigr)^{-1}\hat{V}(\tau_C,S_{\tau_C}) \Gamma_t^{\tau_C}  \ind_{\{t < \tau_C < \tau_I \wedge T; \hat{V}(\tau_C,S_{\tau_C}) \geq 0\}}  \, \Bigr\vert \, \mathcal{G}_t \Bigr] \nonumber \\
    &\phantom{=} + \mathbb{E}^{\Qxx} \Bigl[ \Bigl. \alpha \bigl(r_f -r_c\bigr) \int_{t \wedge \tau}^\tau \bigl(B_s^{r_f}\bigr)^{-1} \hat{V}(s,S_s) \Gamma_t^s \, ds \, \Bigr\vert \, \mathcal{G}_t \Bigr].
  \label{eq:reprV}
  \end{align}
This decomposition allows for the nice interpretation of the XVA in terms of four separate contributing terms. The first term corresponds to the default- and collateralization free price under funding constraints, the second one to the (funding-adjusted) payout after  default of the trader, the third one to the (funding-adjusted) payout after counterparty's default and finally the last one to the funding costs of the collateralization procedure.

\subsubsection{Explicit computation of XVA} \label{sec:explicit}

To facilitate the calculation of $\tilde{V}_t$, we prove first a useful lemma.

\begin{lemma}\label{stoppro}
For $t, \, s \in[0,T]$, $ t \leq s$, we have
  \begin{equation}
    \mathbb{E}^\Qxx\bigl[\bigl. \ind_{\{\tau \geq s\}} \, \bigr\vert \, \mathcal{G}_t \bigr] = e^{-(\hIQ+\hCQ)(s-t)}  \ind_{\{\tau \geq t\}}  \label{eq:probabs1}
  \end{equation}
Moreover, if $(M_t)_{0 \leq t \leq T}$ is an $\mathbb{F}$-measurable, square integrable $(\Qxx,\mathbb{G})$-martingale and $\lambda \in \mathbb{R}$, then if $\lambda \neq \hIQ$, $\lambda \neq \hCQ$ and $\lambda \neq \hIQ+\hCQ$, we have
  \begin{align}
    &\phantom{==} \mathbb{E}^{\Qxx} \bigl[ \bigl. e^{\lambda \tau_I} M_{\tau_I} \ind_{\{t < \tau_I < \tau_C \wedge T\}} \, \bigr\vert \, \mathcal{G}_t \bigr]
    \label{eq:probabs2} \\
    & = \frac{\hIQ}{\lambda-\hIQ} e^{\lambda t} \biggl(\frac{\hCQ}{\lambda-\hIQ-\hCQ} \Bigl(e^{(\lambda-\hIQ-\hCQ){(T-t)}}-1\Bigr)- 1 + e^{-\hCQ{(T-t)}} e^{(\lambda-\hIQ)(T-t)}\biggr)M_t \ind_{\{\tau > t\}} \nonumber,
  \end{align}

as well as the analogous statement to \eqref{eq:probabs2} with $I$ and $C$ interchanged.
\end{lemma}

\begin{proof}
All the calculations follow from the fact that the indicators are independent of the filtration $\mathbb{F}$, the memorylessness of the exponential distribution as well as its explicit distributional properties and the independence of the two default times. We get by the ``key lemma'' \cite[Lemma 5.1.2]{bielecki01}
  \begin{align*}
    \mathbb{E}^\Qxx\bigl[\bigl. \ind_{\{\tau \geq s\}} \, \bigr\vert \, \mathcal{G}_t \bigr] & = \frac{\mathbb{E}^\Qxx\bigl[\bigl. \ind_{\{\tau \geq s\}} \, \bigr\vert \, \mathcal{F}_t \bigr]}{\Qxx[\tau>t]} \ind_{\{\tau \geq t\}} = \frac{\Qxx[\tau >s]}{\Qxx[\tau>t]} \ind_{\{\tau \geq t\}} = e^{-(\hIQ+\hCQ) (s-t)}\ind_{\{\tau \geq t\}}.
  \end{align*}
Moreover, we note that we have again by the ``key lemma'' and memorylessness of the exponential distribution
  \begin{align*}
    &\phantom{==} \mathbb{E}^{\Qxx} \bigl[ \bigl. e^{\lambda \tau_I} M_{\tau_I} \ind_{\{t < \tau_I < \tau_C \wedge T\}} \, \bigr\vert \, \mathcal{G}_t \bigr]  = \frac{\mathbb{E}^{\Qxx} \bigl[ \bigl. e^{\lambda \tau_I} M_{\tau_I} \ind_{\{t < \tau_I < \tau_C \wedge T\}} \, \bigr\vert \, \mathcal{F}_t \bigr]}{\Qxx[\tau>t]} \ind_{\{\tau>t\}}\\
    &= \frac{\mathbb{E}^{\Qxx} \bigl[ \bigl. \int_t^\infty \int_t^{z \wedge T} e^{\lambda y} M_y \hIQ e^{-\hIQ y}\hCQ e^{-\hCQ z} \, dydz \, \bigr\vert \, \mathcal{F}_t \bigr]}{\Qxx[\tau>t]} \ind_{\{\tau>t\}} \\
    &= \int_0^\infty \int_0^{z \wedge (T-t)} e^{\lambda (y+t)} \hIQ e^{-\hIQ y}\hCQ e^{-\hCQ z} \, dydz \, M_t \ind_{\{\tau>t\}}\\
    & = \frac{\hIQ}{\lambda-\hIQ} e^{\lambda t} \biggl(\frac{\hCQ}{\lambda-\hIQ-\hCQ} \Bigl(e^{(\lambda-\hIQ-\hCQ){(T-t)}}-1\Bigr)- 1 + e^{-\hCQ{(T-t)}} e^{(\lambda-\hIQ)(T-t)}\biggr)M_t \ind_{\{\tau > t\}}.
  \end{align*}
The reverse statement follows in the same way.
\end{proof}

We now proceed to the calculation of $\tilde{V}_t$, term by term. First we get by the tower property, independence of the filtration $\mathbb{F}$ and $\mathbb{H}$ and \eqref{eq:probabs1} that
  \begin{align*}
    & \phantom{=} \mathbb{E}^{\Qxx} \Bigl[ \Bigl.  \bigl(B_T^{r_f}\bigr)^{-1} \Phi(S_T) \ind_{\{\tau = T\}}  \Gamma_t^\tau  \, \Bigr\vert \, \mathcal{G}_t \Bigr] \\
    &= \frac{B_T^{r_D}}{B_T^{r_f}} e^{(r_f - r_D)(T - t )}e^{(r_f - r_D)(T - t )} \mathbb{E}^{\Qxx} \biggl[ \biggl. \mathbb{E}^{\Qxx} \Bigl[ \Bigl.  \bigl(B_T^{r_D}\bigr)^{-1} \Phi(S_T) \ind_{\{\tau = T\}} \, \Bigr\vert \, \mathcal{F}_t \vee \mathcal{H}_T\Bigr] \, \biggr\vert \, \mathcal{G}_t \biggr]\\
    &= \frac{B_T^{r_D}}{B_T^{r_f}} e^{(r_f - r_D)(T - t )}e^{(r_f - r_D)(T - t )}\mathbb{E}^{\Qxx} \bigl[ \bigl.\ind_{\{\tau = T\}} \, \bigr\vert \, \mathcal{G}_t \bigr] (B_t^{r_D})^{-1} \hat{V}(t,S_t)\\
    &= \frac{B_T^{r_D}}{B_T^{r_f}} \frac{1}{B_t^{r_D}} e^{(r_f - r_D -\hIQ)(T - t )}e^{(r_f - r_D - \hCQ)(T - t )} \hat{V}(t,S_t)\ind_{\{\tau \geq t\}}
  \end{align*}

To simplify the calculation of the $\CVA$ and $\DVA$ terms, we make the following definitions:
  \begin{align*}
    \hat{\Theta}_t^{s,+} & := \mathbb{E}^{\Qxx}\Bigl[ \Bigl. e^{-r_D (s-t)} \hat{V}(s,S_s) \ind_{\{\hat{V}(s,S_s) \geq 0\}}\, \Bigr\vert \, \mathcal{F}_t \Bigr],\\
    \hat{\Theta}_t^{s,-} & := \mathbb{E}^{\Qxx}\Bigl[ \Bigl. e^{-r_D (s-t)} \hat{V}(s,S_s) \ind_{\{\hat{V}(s,S_s) < 0\}}\, \Bigr\vert \, \mathcal{F}_t \Bigr].
  \end{align*}
We note that $\hat{\Theta}_t^{s,+}$ (respectively $\hat{\Theta}_t^{s,-}$) are public time-$t$-values of compound options, and $\bigl(B_t^{r_D}\bigr)^{-1}\hat{\Theta}_t^{s,+}$, $\bigl(B_t^{r_D}\bigr)^{-1}\hat{\Theta}_t^{s,-}$ are the discounted values (and $\mathbb{F}$-measurable $\mathbb{G}$-martingales). More precisely, $\hat{\Theta}_t^{s,+}$ is the value of a call option on the option $\Phi(S_T)$ with maturity $s$ and strike zero, and $\hat{\Theta}_t^{s,-}$ is the respective put option. We have
  \[
    \hat{V}(\tau_I, S_{\tau_I}) \ind_{\{ \hat{V}(\tau_I, S_{\tau_I}) >0\}} = \Exx^\Qxx \Bigl[ \Bigl. \hat{V}(\tau_I, S_{\tau_I}) \ind_{\{ \hat{V}(\tau_I, S_{\tau_I}) >0\}} \, \Bigr\vert \mathcal{F}_{\tau_I} \Bigr] = \Theta_{\tau_I}^{\tau_I,+}
  \]
Note that using Lemma \ref{stoppro}, we get for the first part of the $\DVA$ type term
  \begin{align*}
    & \phantom{==}\mathbb{E}^{\Qxx} \Bigl[ \Bigl.   \bigl(B_{\tau_I}^{r_f}\bigr)^{-1}\bigl(1-(1- \alpha) L_I\bigr)\hat{V}(\tau_I,S_{\tau_I}) \Gamma_t^{\tau_I}  \ind_{\{t \leq \tau_I < \tau_C \wedge T; \hat{V}(\tau_I,S_{\tau_I}) \geq 0\}} \, \Bigr\vert \, \mathcal{G}_t \Bigr]\\
    & = \bigl(1-(1- \alpha) L_I\bigr) e^{-( 2( r_f - r_D)t} \Bigl(1+ \frac{r_D-r_f}{\hIQ}\Bigr) \cdot \mathbb{E}^{\Qxx} \Bigl[ \Bigl.   e^{\lambda \tau_I} \bigl(B_{\tau_I}^{r_D}\bigr)^{-1} \hat{\Theta}^{\tau_I,+}_{\tau_I} \ind_{\{t \leq \tau_I < \tau_C \wedge T\}} \, \Bigr\vert \, \mathcal{G}_t \Bigr]\\
    & = \bigl(1-(1- \alpha) L_I\bigr)e^{-2( r_f - r_D)t} \Bigl(1+ \frac{r_D-r_f}{\hIQ}\Bigr)\frac{\hIQ}{\lambda-\hIQ} e^{\lambda t} \\
    & \phantom{=} \cdot \biggl(\frac{\hCQ}{\lambda-\hCQ-\hIQ} \Bigl(e^{(\lambda-\hCQ-\hIQ){(T-t)}}-1\Bigr)- 1 + e^{-\hCQ{(T-t)}} e^{(\lambda-\hIQ)(T-t)}\biggr)\frac{1}{B_t^{r_D}}\hat{\Theta}^{\tau_I,+}_t \ind_{\{\tau \geq t\}} \\
    & = \bigl(1-(1- \alpha) L_I\bigr)e^{(r_D - r_f)t} \Bigl(1+ \frac{r_D-r_f}{\hIQ}\Bigr)\frac{\hIQ}{\lambda-\hIQ} \\
    & \phantom{=} \cdot \biggl(\frac{\hCQ}{\lambda-\hCQ-\hIQ} \Bigl(e^{(\lambda-\hCQ-\hIQ){(T-t)}}-1\Bigr)- 1 + e^{-\hCQ{(T-t)}} e^{(\lambda-\hIQ)(T-t)}\biggr)\frac{1}{B_t^{r_D}}\hat{\Theta}^{\tau_I,+}_t \ind_{\{\tau \geq t\}}
  \end{align*}
with $\lambda :=  r_f - r_D$. Notice that by Assumption \ref{ass:necessary} we have
  \[
    \lambda = r_f-r_D < r^i + h_i^{\Px} - r_D = h_i^{\Qxx}
  \]
for $i \in \{ I, C\}$, and thus the conditions of Lemma \ref{stoppro} are satisfied. In the second equality above, we have used that $\Theta_{\tau_I}^{\tau_I,+} = \hat{V}(\tau_I,S_{\tau_I}) \ind_{\{\hat{V}(\tau_I,S_{\tau_I})>0\}}$ along with the fact that the process $\bigl(B_{s}^{r_D}\bigr)^{-1} \Theta_{s}^{\tau_I,+}$ is a $\mathcal{G}_t$ martingale, implying that $\mathbb{E}^{\Qxx} \Bigl[\bigl(B_{\tau_I}^{r_D}\bigr)^{-1} \Theta_{\tau_I}^{\tau_I,+} \bigg| \mathcal{G}_t\Bigr] = \bigl(B_{t}^{r_D}\bigr)^{-1} \Theta_{t}^{\tau_I,+}$. Similarly, we have for the second part
  \begin{align*}
    & \phantom{=}  \mathbb{E}^{\Qxx} \Bigl[\bigl(B_{\tau_I}^{r_f}\bigr)^{-1}\hat{V}(\tau_I,S_{\tau_I}) \Gamma_t^{\tau_I}  \ind_{\{t \leq \tau_I < \tau_C \wedge T; \hat{V}(\tau_I,S_{\tau_I}) < 0\}}  \, \Bigr\vert \, \mathcal{G}_t \Bigr]  = \\
    & = e^{(r_D - r_f)t}  \Bigl(1+ \frac{r_D-r_f}{\hCQ}\Bigr)  \frac{\hIQ}{\lambda-\hIQ}  \\
    & \phantom{=} \cdot \biggl(\frac{\hCQ}{\lambda-\hCQ-\hIQ} \Bigl(e^{(\lambda-\hCQ-\hIQ){(T-t)}}-1\Bigr)- 1 + e^{-\hCQ{(T-t)}} e^{(\lambda-\hIQ)(T-t)}\biggr)\frac{1}{B_t^{r_D}}\hat{\Theta}^{\tau_I,-}_t \ind_{\{\tau \geq t\}}.
  \end{align*}
For the $\CVA$ type term, we obtain
  \begin{align*}
    & \phantom{=}\mathbb{E}^{\Qxx} \Bigl[ \Bigl.   \bigl(B_{\tau_C}^{r_f}\bigr)^{-1}\bigl(1-(1- \alpha) L_C\bigr)\hat{V}(\tau_C,S_{\tau_C}) \Gamma_{t}^{\tau_C}  \ind_{\{t \leq \tau_C < \tau_I \wedge T; \hat{V}(\tau_C,S_{\tau_C}) < 0\}}  \, \Bigr\vert \, \mathcal{G}_t \Bigr] \\
    & = \bigl(1-(1- \alpha) L_C\bigr) e^{(r_D - r_f)t} \Bigl(1+ \frac{r_D-r_f}{\hCQ}\Bigr) \frac{\hCQ}{\lambda-\hCQ} \\
    & \phantom{=} \cdot \biggl(\frac{\hIQ}{\lambda-\hCQ-\hIQ}  \Bigl(e^{(\lambda-\hCQ-\hIQ){(T-t)}}-1\Bigr)- 1 + e^{-\hIQ{(T-t)}} e^{(\lambda-\hCQ)(T-t)}\biggr)\frac{1}{B_t^{r_D}}\hat{\Theta}^{\tau_C,-}_t \ind_{\{\tau \geq t\}}
  \end{align*}
and
  \begin{align*}
    & \phantom{=} \mathbb{E}^{\Qxx} \Bigl[ \bigl(B_{\tau_C}^{r_f}\bigr)^{-1}\hat{V}(\tau_C,S_{\tau_C}) \Gamma_{t}^{\tau_C}  \ind_{\{t \leq \tau_C < \tau_I \wedge T; \hat{V}(\tau_I,S_{\tau_I}) \geq 0\}}  \, \Bigr\vert \, \mathcal{G}_t \Bigr] \\
    & = e^{(r_D - r_f)t} \Bigl(1+ \frac{r_D-r_f}{\hCQ}\Bigr) \frac{\hCQ}{\lambda-\hCQ} \\
    & \phantom{=} \cdot \biggl(\frac{\hIQ}{\lambda-\hCQ-\hIQ} \Bigl(e^{(\lambda-\hCQ-\hIQ){(T-t)}}-1\Bigr)- 1 + e^{-\hIQ{(T-t)}} e^{(\lambda-\hCQ)(T-t)}\biggr)\frac{1}{B_t^{r_D}}\hat{\Theta}^{\tau_C,+}_t \ind_{\{\tau \geq t\}}
  \end{align*}

Finally, the funding costs of the collateralization procedure are given by
  \begin{align*}
    &\phantom{==} \mathbb{E}^{\Qxx} \Bigl[ \Bigl. \int_{t \wedge \tau}^\tau \bigl(r_c -r_f\bigr) \tilde{C}_s \Gamma_t^s \, ds \, \Bigr\vert \, \mathcal{G}_t \Bigr] \\
    & = \alpha \bigl(r_f -r_c\bigr) \int_t^T \mathbb{E}^{\Qxx} \Bigl[ \Bigl.  \bigl(B_s^{r_f}\bigr)^{-1} \hat{V}(s,S_s) \Gamma_t^s \ind_{\{ \tau \geq s\}} \, \Bigr\vert \, \mathcal{G}_t \Bigr] \, ds \\
    & = \alpha \bigl(r_f -r_c\bigr) \int_t^T  \frac{B_s^{r_D} }{B_s^{r_f} } e^{(r_f-r_D)(s-t)}  e^{(r_f-r_D)(s-t)}  \bigl(B_t^{r_D}\bigr)^{-1} \hat{V}(t,S_t) \mathbb{E}^{\Qxx} \bigl[ \bigl. \ind_{\{ \tau \geq s\}} \, \bigr\vert \, \mathcal{G}_t \bigr] \, ds \\
    & = \alpha \bigl(r_f -r_c\bigr) \bigl(B_t^{r_D}\bigr)^{-1} \hat{V}(t,S_t) \int_t^T  \frac{B_s^{r_D} }{B_s^{r_f} } e^{(r_f-r_D-\hIQ)(s-t)}  e^{(r_f-r_D-\hCQ)(s-t)}  \, ds \ind_{\{\tau \geq t\}}  \\
    & =  \alpha\frac{r_f -r_c}{\hIQ + \hCQ - \lambda}\frac{1}{B_t^{r_D} } \biggl( \frac{B_t^{r_D} }{B_t^{r_f}} - \frac{B_{T}^{r_D} }{B_T^{r_f}} e^{(r_f-r_D-\hIQ) (T-t)} e^{ (r_f-r_D-\hCQ )(T-t)} \biggr) \hat{V}(t,S_t) \ind_{\{ \tau \geq t\}}.
  \end{align*}
The situation simplifies considerably if the payoff $\Phi$ is non-negative, as for short call or put positions, or non-positive as for long ones.
In the case of a non-negative payoff we have
  \[
    \hat{\Theta}^{\tau_C,+}_t = \hat{\Theta}^{\tau_I,+}_t = \hat{V}(t,S_t), \qquad \hat{\Theta}^{\tau_C,-}_t = \hat{\Theta}^{\tau_I,-}_t = 0.
  \]
Recalling that $V_t=B_t^{r_f} \tilde{V}_t$, we obtain
  \begin{align*}
    \XVA_t & = \Biggl(\frac{B_t^{r_f}}{B_t^{r_D}} \biggl(\frac{B_T^{r_D}}{B_T^{r_f}} e^{(r_f-r_D-\hIQ)(T-t)}e^{(r_f-r_D-\hCQ)(T-t)} \\
    & \phantom{=} + \bigl(1-(1- \alpha) L_I\bigr)e^{({r_D - r_f})t} \Bigl(1+ \frac{r_D-r_f}{\hIQ}\Bigr) \frac{\hIQ}{\lambda-\hIQ} \\
    & \phantom{=} \cdot \Bigl(\frac{\hCQ}{\lambda-\hCQ-\hIQ} \bigl(e^{(\lambda-\hCQ-\hIQ){(T-t)}}-1\bigr)- 1 + e^{-\hCQ{(T-t)}} e^{(\lambda-\hIQ)(T-t)}\Bigr)\\
    & \phantom{=} + e^{({r_D - r_f})t} \Bigl(1+ \frac{r_D-r_f}{\hCQ}\Bigr) \frac{\hCQ}{\lambda-\hCQ} \\
    & \phantom{=} \cdot \Bigl(\frac{\hIQ}{\lambda-\hCQ-\hIQ} \bigl(e^{(\lambda-\hCQ-\hIQ){(T-t)}}-1\bigr)- 1 + e^{-\hIQ{(T-t)}} e^{(\lambda-\hCQ)(T-t)}\Bigr)\\
    & \phantom{=} + \alpha\frac{r_f -r_c}{\hCQ + \hIQ - \lambda} \Bigl( \frac{B_t^{r_D} }{B_t^{r_f}} - \frac{B_{T}^{r_D} }{B_T^{r_f}} e^{(r_f-r_D-\hCQ) (T-t)} e^{ (r_f-r_D-\hIQ )(T-t)} \Bigr)  \biggr) - 1\Biggr) \hat{V}(t,S_t) \ind_{\{\tau \geq t\}} \\
    & :=  ( A - 1) \hat{V}(t,S_t) \ind_{\{\tau \geq t\}},
  \end{align*}
where $A$ is short-hand notation for the adjustment factor, i.e. $A = V_t /\hat{V}(t,S_t)$. Moreover, applying Theorem \ref{thm:strat} we obtain that on the set $\{t<\tau\}$, the optimal stock and bond investment strategies are given by
  \begin{align*}
    \xi_t &= A \times \hat{V}_S(t,S_t) ,\\
    \xi_{t}^i &=  \frac{A \times \hat{V}(t,S_t)- \theta_i(\hat v(t,S_t))}{e^{-(r_D + h_i^{\Qxx}) (T-t)}} ,~i\in\{I,C\}.
  \end{align*}
where we have used that on  $\{t < \tau \}$ the bond price is $P_{t}^i = e^{-(r_D + h_i^{\Qxx}) (T-t)}$ and used, with slight abuse of notation, $\theta_{i},~i\in\{I,C\}$,as the contract at default of either party, similar to the way it was defined in \eqref{eq:theta}. Specifically,
  \begin{align*}
    \theta_{C}(\hat v) & := \hat v +   L_C ((1-\alpha) \hat{v} )^{-},  \\
    \theta_{I}(\hat v) & := \hat v - L_I  ((1-\alpha)\hat{v} )^{+}.
  \end{align*}
This can be verified by a direct application of the ``key lemma'' \cite[Lemma 5.1.2]{bielecki01} to the bond price in \eqref{eq:priceproc}.

Figure \ref{fig:pricedec1} plots the different terms in the decomposition given by \eqref{eq:reprV}. Under the parameter configuration corresponding to a safer scenario (left panel), the predominant contribution is given by the default and collateralization-free price under funding constraints. The contributions coming from the closeout positions realized at default time of the trader and of his counterparty are small. If the default risk of trader and counterparty are very high (right panel), the contribution to the XVA coming from the closeout payoff becomes larger. A direct comparison of the bottom panels of figures \ref{fig:defXVA} and \ref{fig:defXVArisk} shows that a similar number of trader{'}s bond shares are used to replicate the jumps to the closeout values. However, as the return on the trader{'}s bond under the valuation measure is higher under the risky scenario, it gives a larger contribution to the XVA.

  \begin{figure}
    \includegraphics[width=0.49\textwidth]{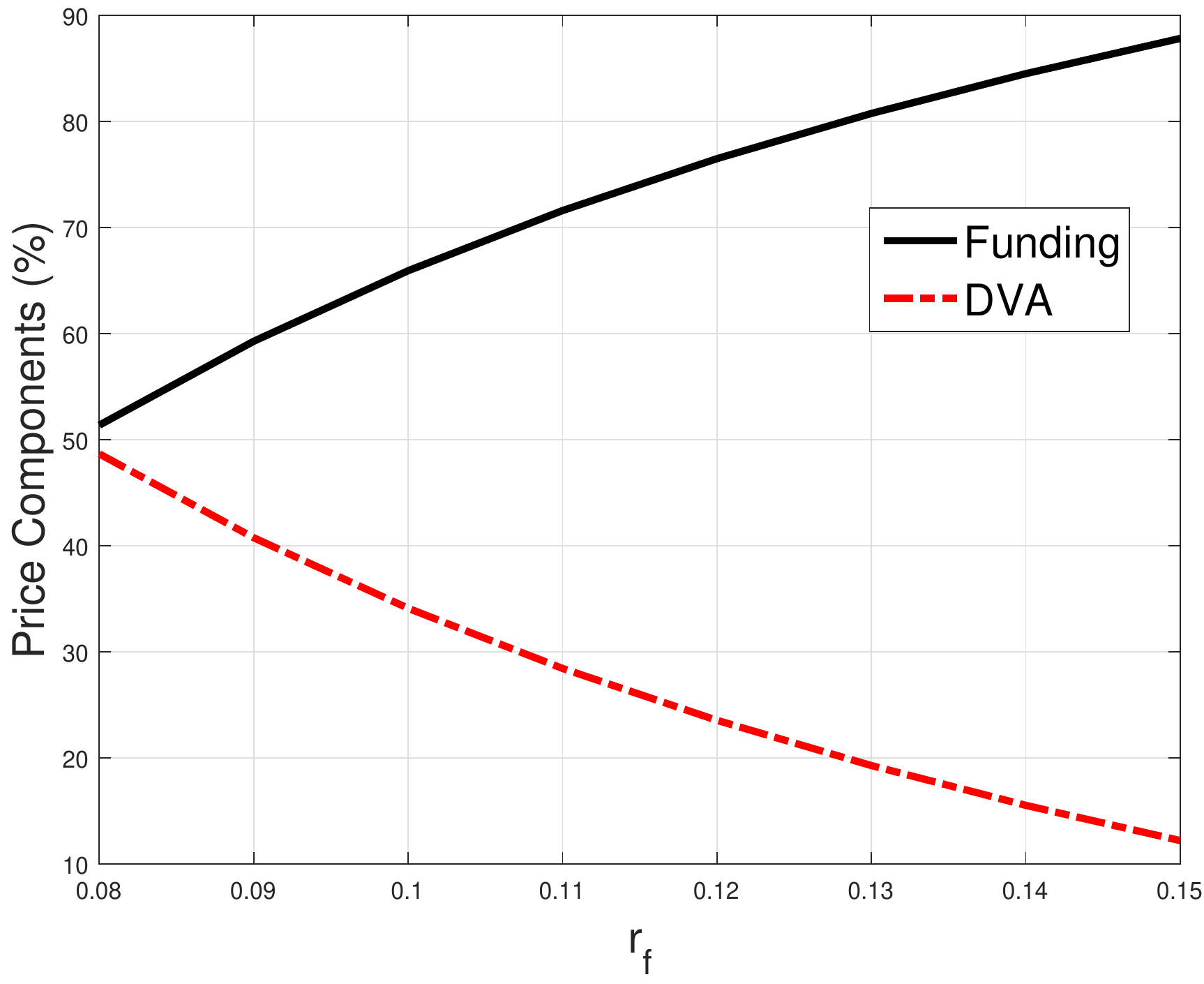}
    \includegraphics[width=0.49\textwidth]{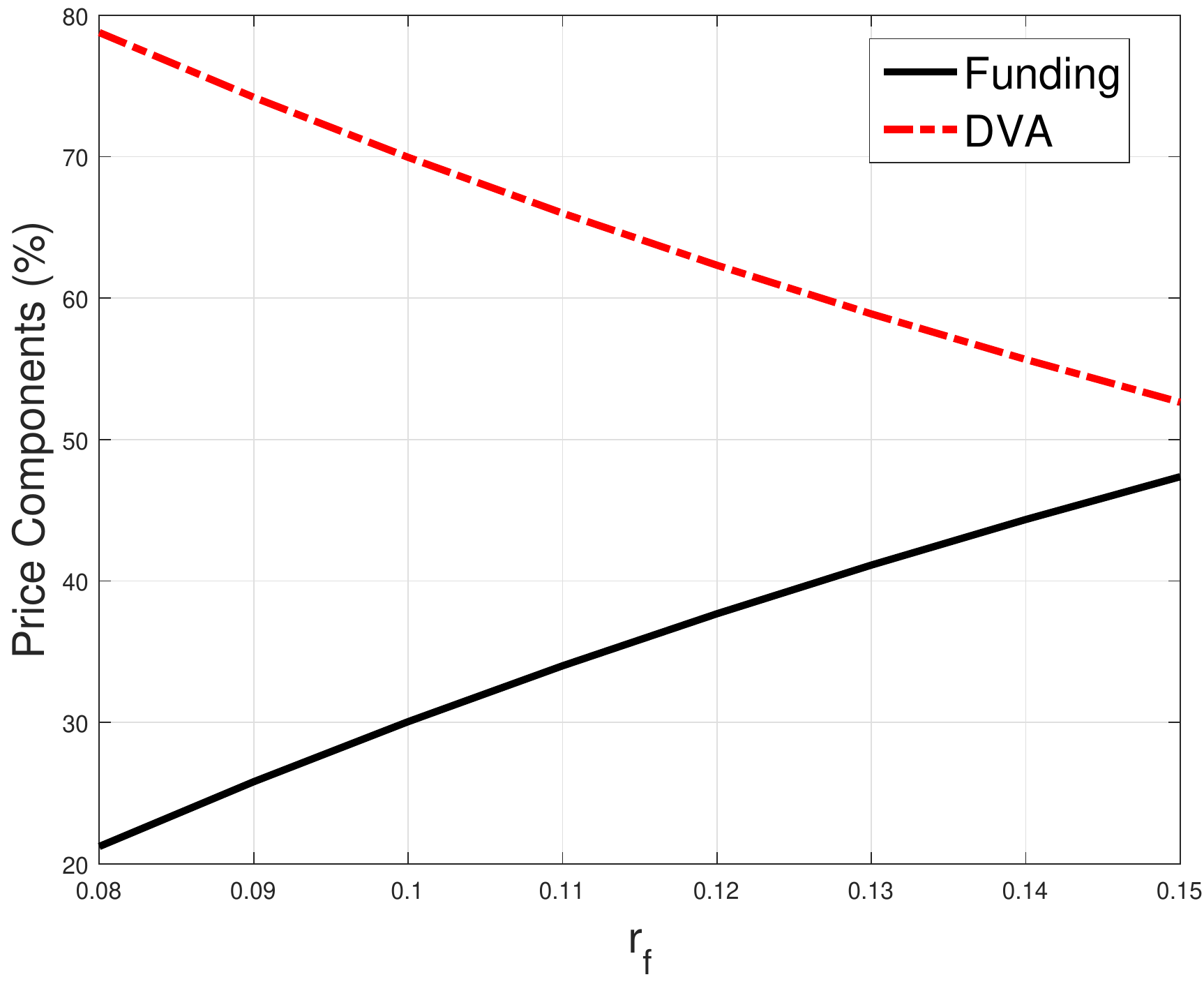}
    \caption{Price decompositions (in \% of the market value) in the Piterbarg model with defaults. We set $r_D = 0.05$, $r_c = 0.01$, $\sigma = 0.2$, $\alpha=0.25$, $L_C = 0.5$, $L_I = 0.5$. The claim is an at-the-money European call option with maturity $T=1$. Left graph: $h_I^{\Qxx} = 0.15$, $h_C^{\Qxx} = 0.2$. Right graph: $h_I^{\Qxx} = 0.5$, $h_C^{\Qxx} = 0.5$.}
  \label{fig:pricedec1}
  \end{figure}

When $\alpha$ is high, the XVA is positive, i.e. $V_0 - \hat{V}(0,S_0) > 0$. The jump to the closeout position when the trader defaults first is given by $V_{\tau_I} - \hat{V}(\tau_I, S_{\tau_I})(1-(1-\alpha)L_I)$. It is higher than the corresponding jump when the counterparty is the first to default, given by $V_{\tau_C} - \hat{V}(\tau_C, S_{\tau_C})$. Consequently, we expect a higher number of trader bond shares to be used in the replication strategy relatively to the number of counterparty bond shares. This is reflected in figures \ref{fig:defXVA} and \ref{fig:defXVArisk}. If $\alpha$ is small, the XVA is negative and the size of the jump to the closeout value at default of the counterparty, $V_{\tau_C} - \hat{V}(\tau_C, S_{\tau_C})$, would also be negative. Hence, to replicate this position the trader would need to short counterparty bonds. However, if the trader were to default first, the size of the jump would be given by $V_{\tau_I} - \hat{V}(\tau_I, S_{\tau_I}) (1-(1-\alpha)L_I)$ and be positive. Consequently, the trader would go long on his own bonds to replicate the jump. As $\alpha$ increases, the size of the jump to the closeout payoff decreases, leading the trader to reduce the position in his own bonds.

A direct comparison of figures \ref{fig:defXVA} and \ref{fig:defXVArisk} suggests that under the risky scenario and for not too high collateralization levels, the trader increases the size of the long position in his own bonds and partly finances it using the proceeds coming from a short position in counterparty bonds. Indeed, when $\alpha = 0$ and $r_f=0.08$ the trader shorts a higher number of counterparty bonds under the risky scenario relatively to the safe scenario. At the same time, he increases the number of shares of his own bonds under the risky scenario.

  \begin{figure}[ht!]
    \centering
      \includegraphics[width=6.6cm]{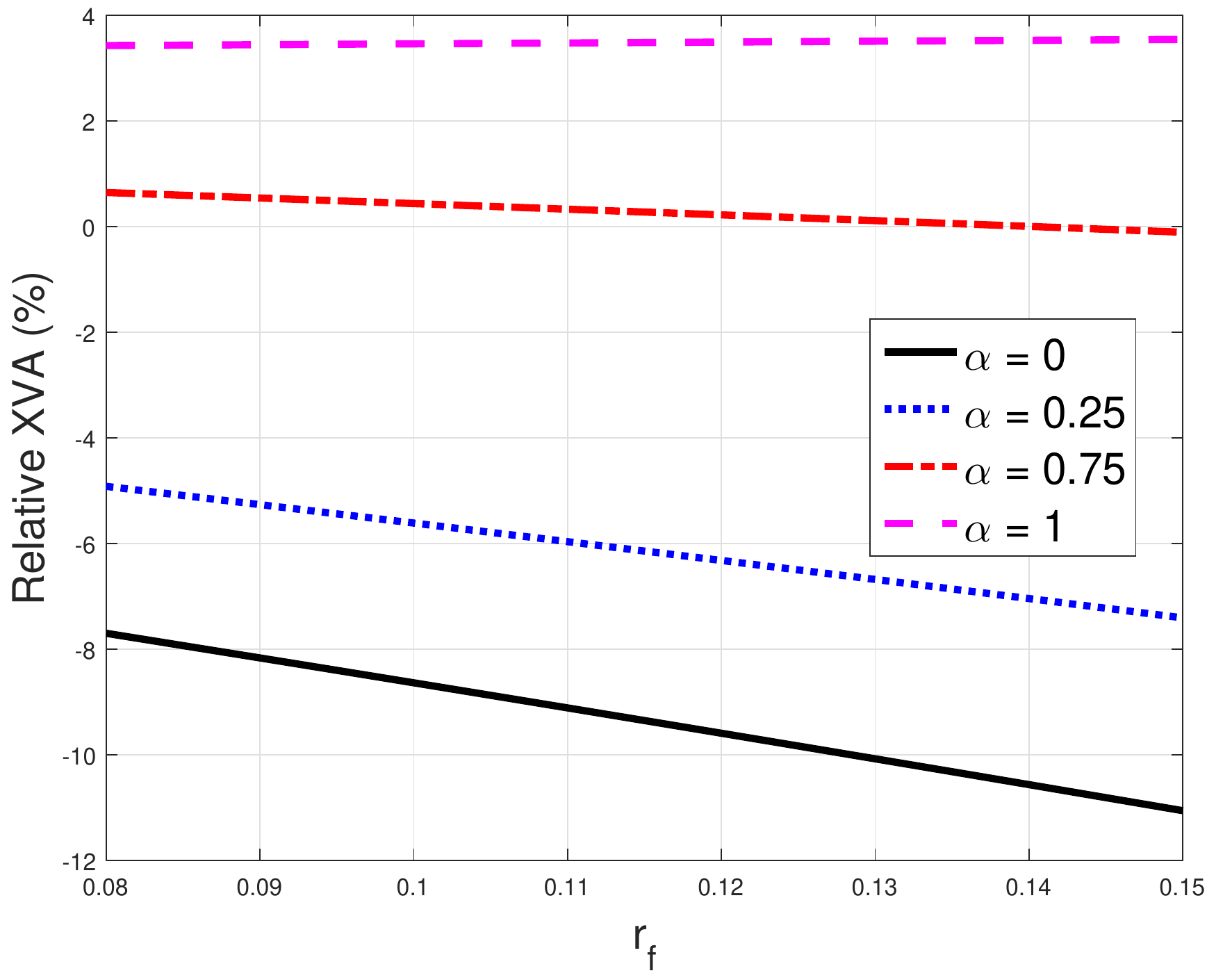}
      \includegraphics[width=6.6cm]{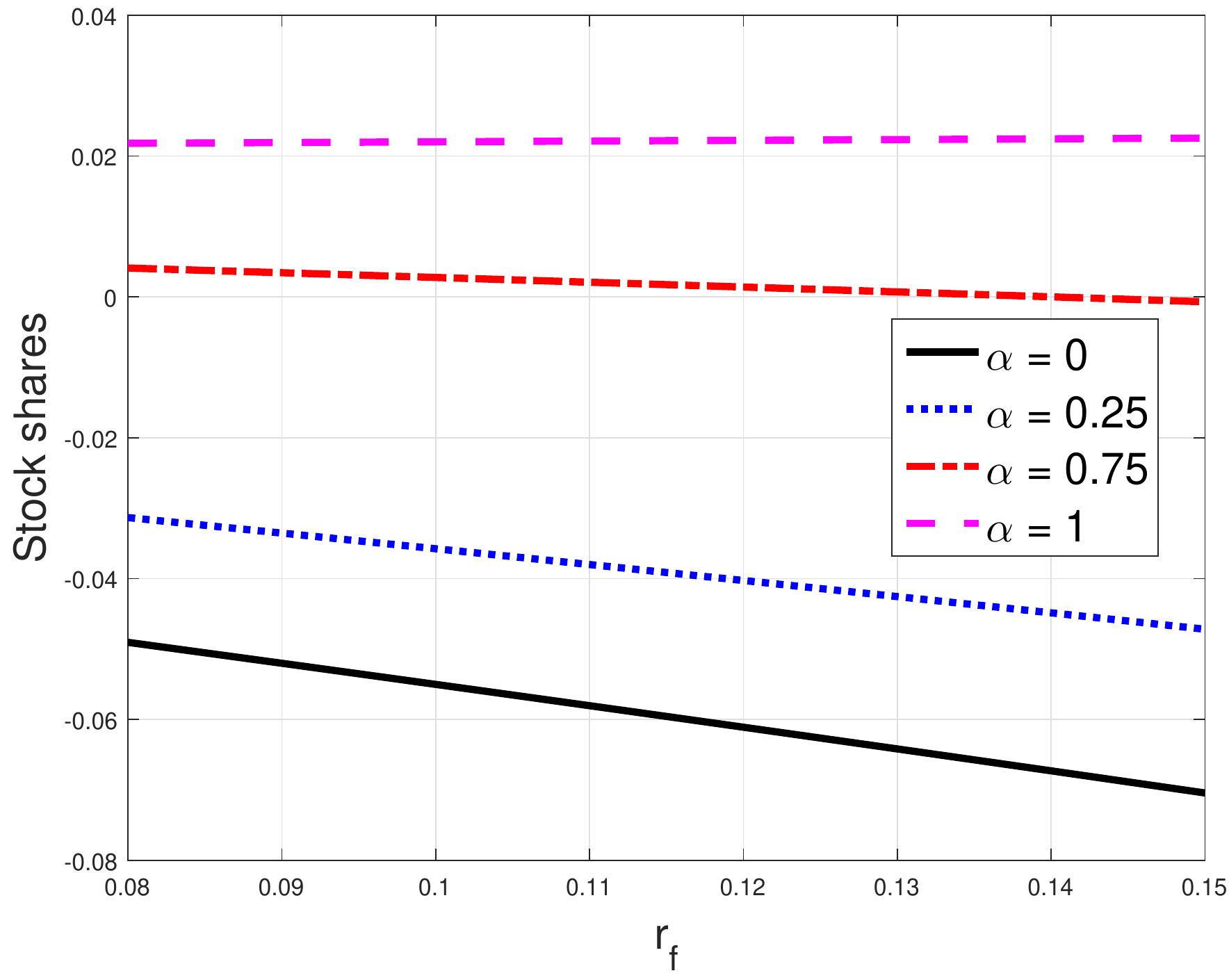}
      \includegraphics[width=6.6cm]{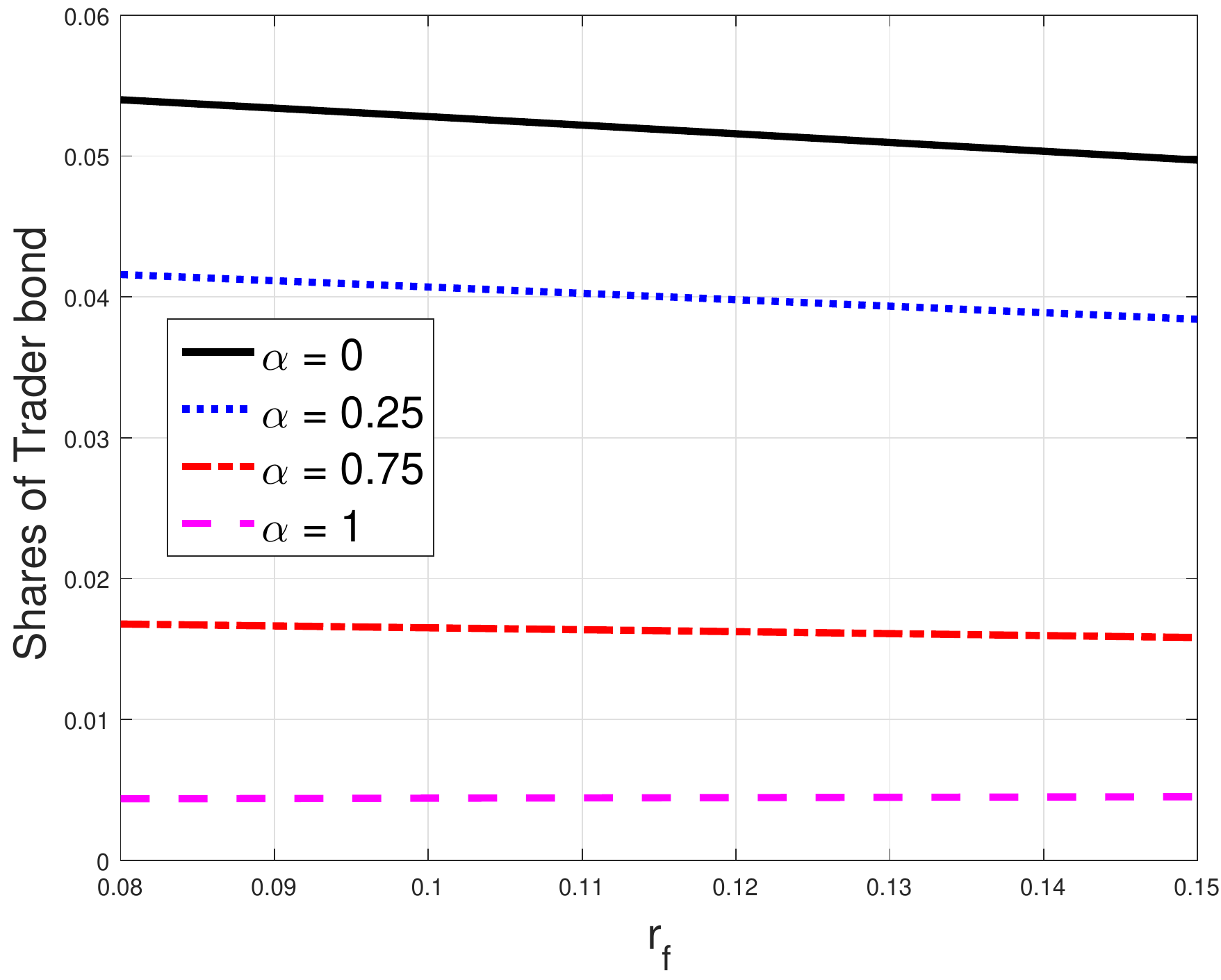}
      \includegraphics[width=6.6cm]{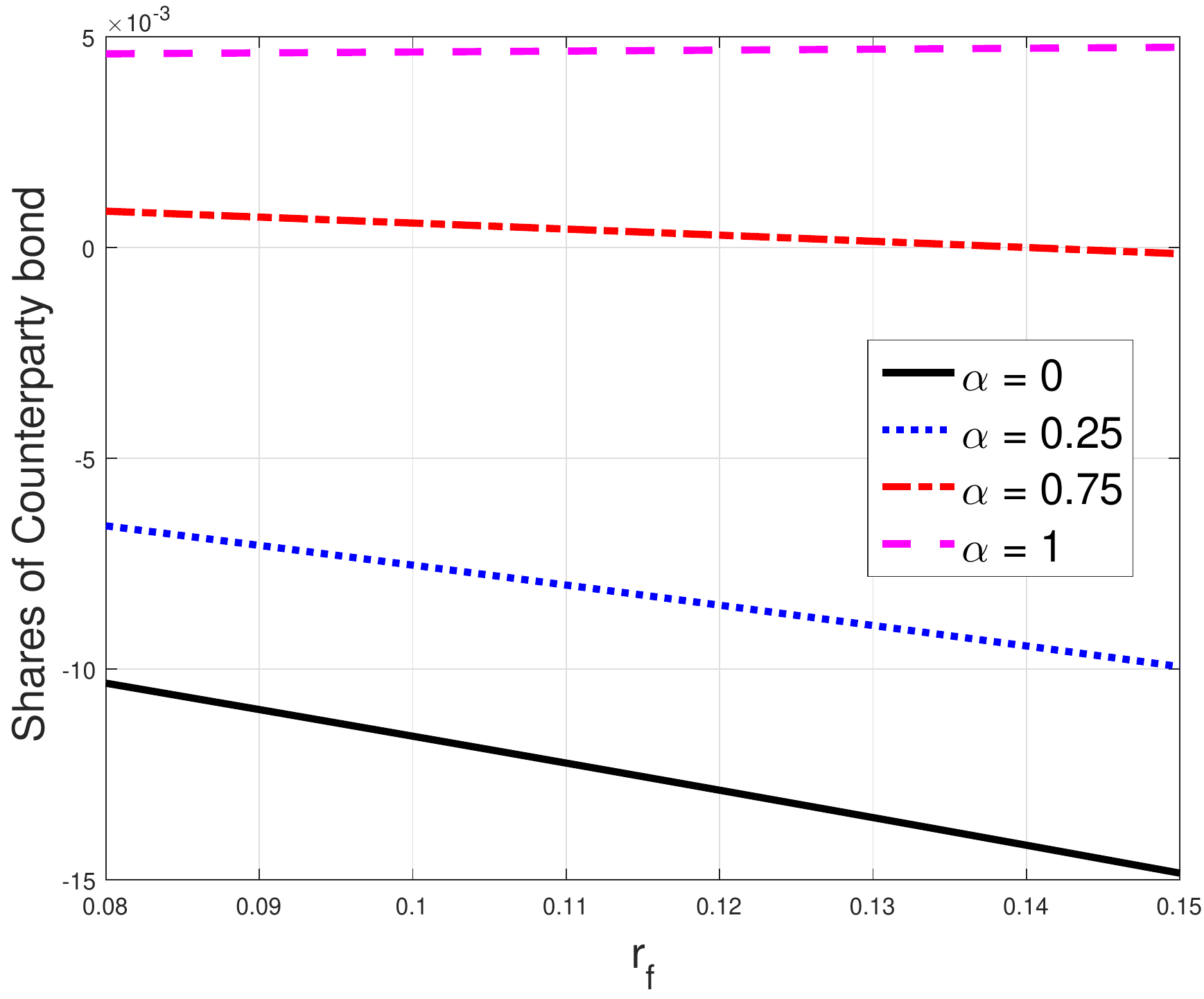}
    \caption{Top left: XVA as a function of $r_f$ for different $\alpha$. Top right: Number of stock shares in the replication strategy. Bottom left: Number of trader bond shares in the replication strategy. Bottom right: Number of counterparty bond shares in the replication strategy. We set $r_D = 0.05$, $r_c = 0.01$, $\sigma = 0.2$, $L_C = 0.5$, $L_I = 0.5$, $h_I^{\Qxx} = 0.15$ and $h_C^{\Qxx} = 0.2$. The claim is an at-the-money European call option with maturity $T=1$.}
  \label{fig:defXVA}
  \end{figure}

  \begin{figure}[ht!]
    \centering
      \includegraphics[width=6.6cm]{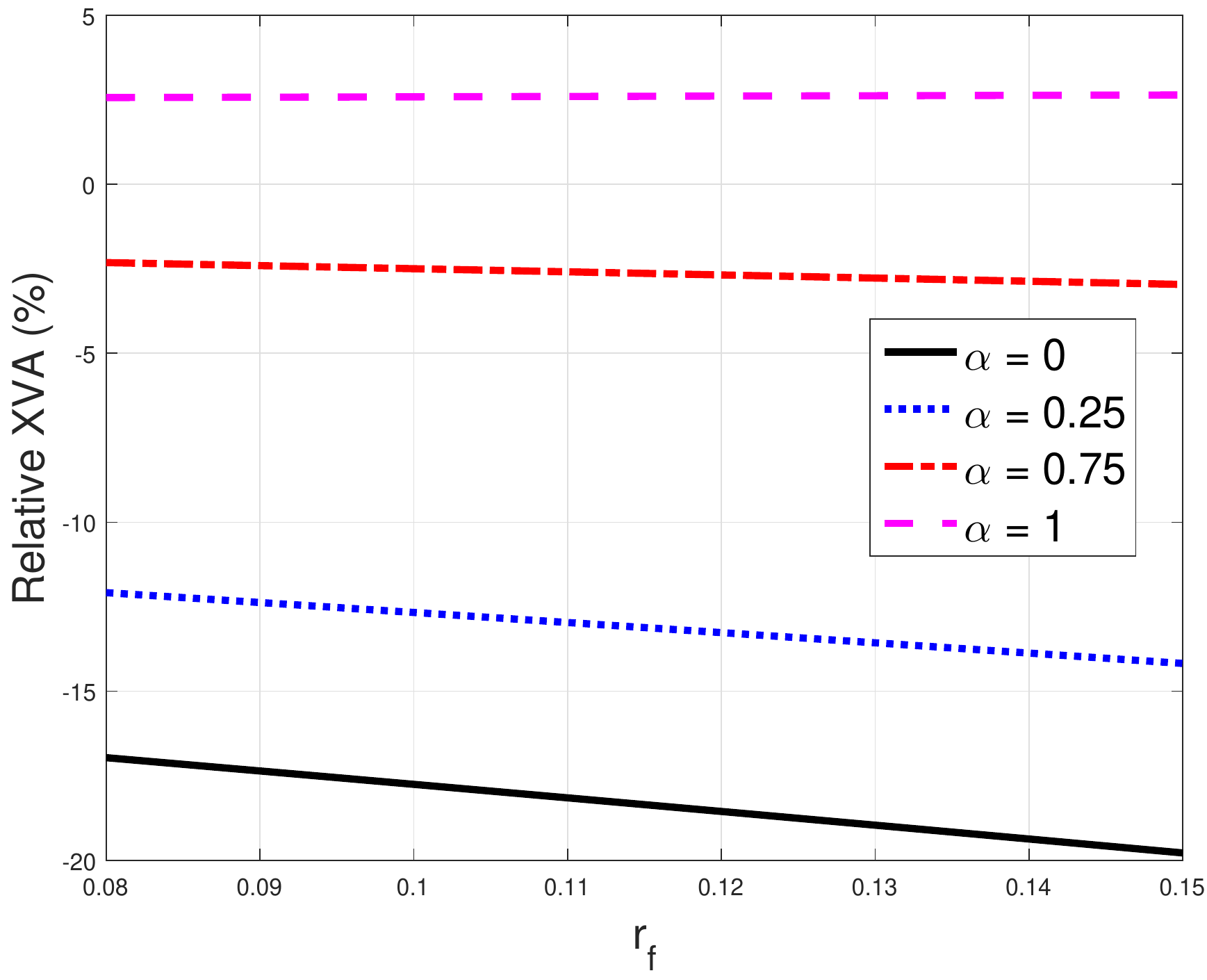}
      \includegraphics[width=6.6cm]{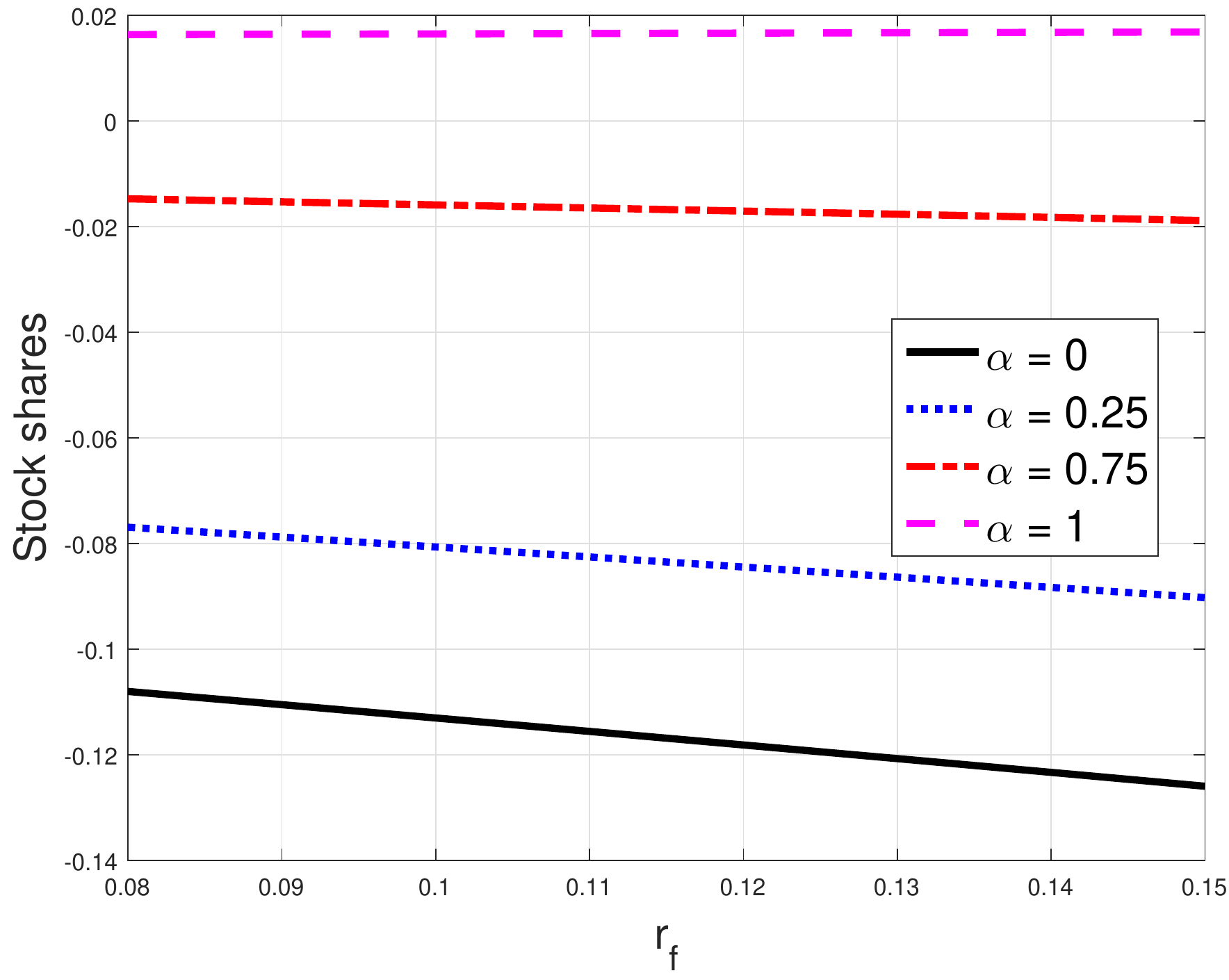}
      \includegraphics[width=6.6cm]{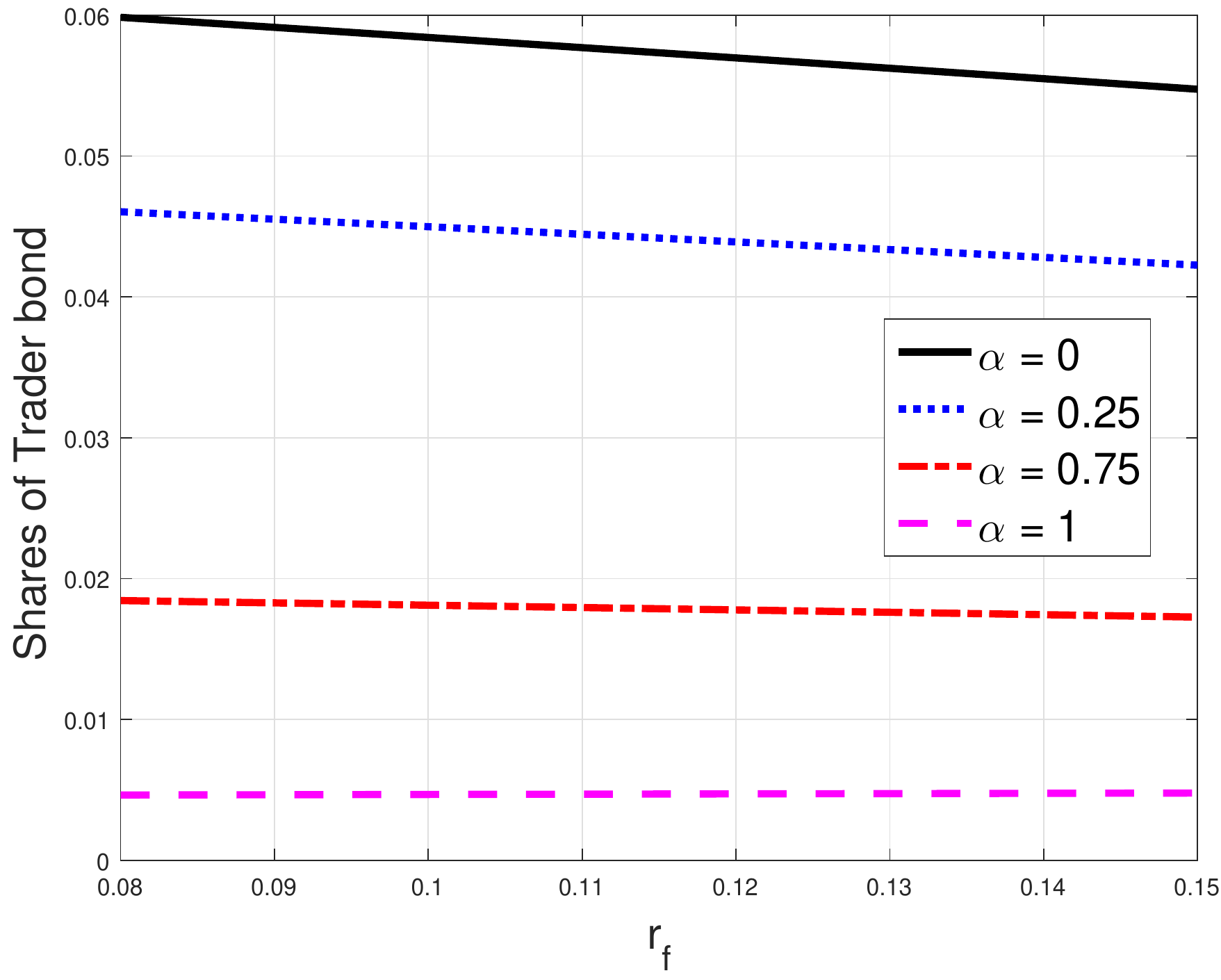}
      \includegraphics[width=6.6cm]{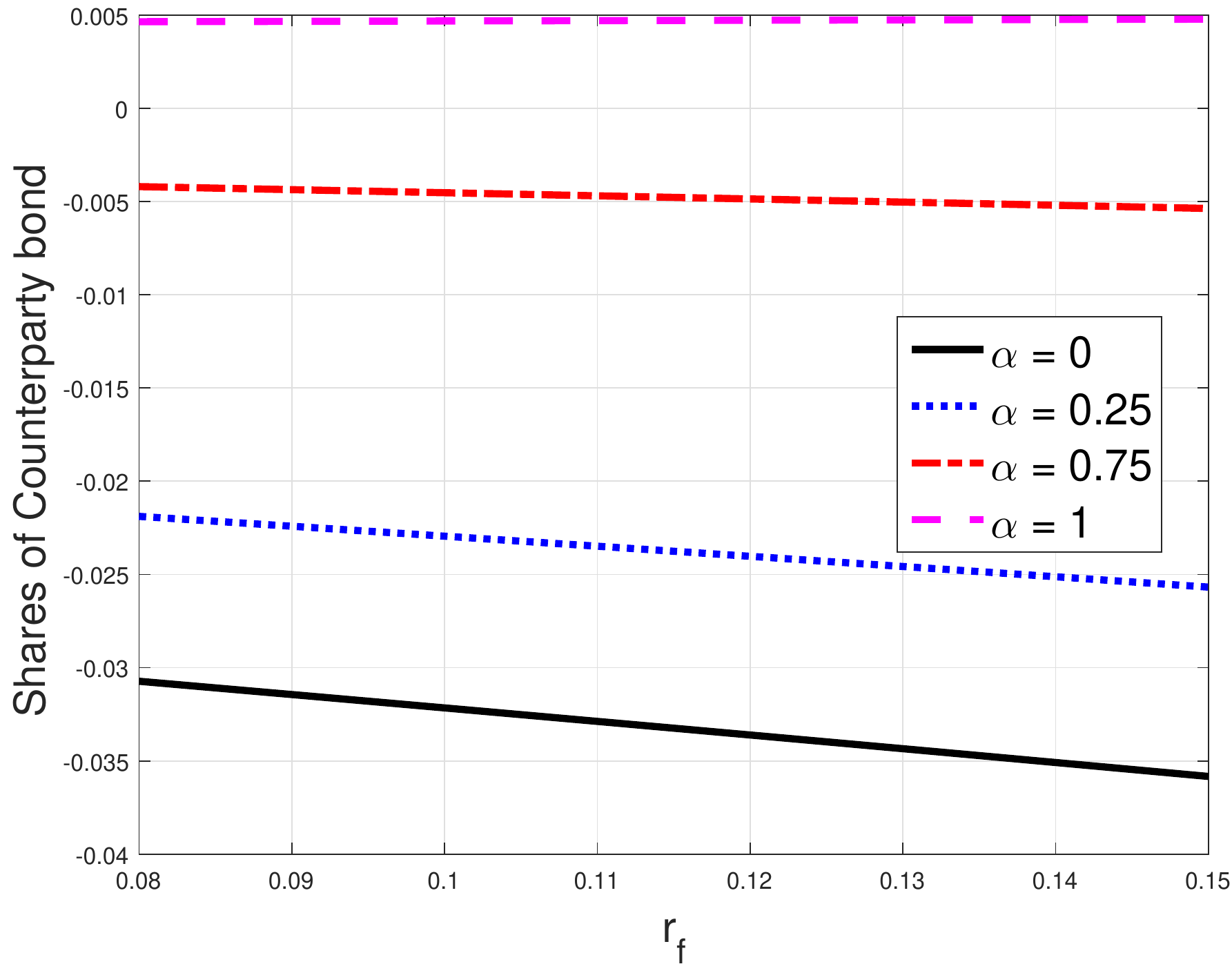}
    \caption{Top left: XVA as a function of $r_f$ for different $\alpha$. Top right: Number of stock shares in the replication strategy. Bottom left: Number of trader bond shares in the replication strategy. Bottom right: Number of counterparty bond shares in the replication strategy. We set $r_D = 0.05$, $r_c = 0.01$, $\sigma = 0.2$, $L_C = 0.5$, $L_I = 0.5$, $h_I^{\Qxx} = 0.5$ and $h_C^{\Qxx} = 0.5$. The claim is an at-the-money European call option with maturity $T=1$.}
  \label{fig:defXVArisk}
  \end{figure}

\subsection{Counterparty adjustments as special cases of XVA} \label{sec:CVAspecial}
We show that bilateral credit valuation adjustments are recovered from our total valuation adjustment formula in the absence of funding costs. This means that all rates are the same, namely $r_D=r_f=r_c$, hence public and private valuations coincide. Under this setting, it holds that $\Gamma_t^s = 1$ for any $0\leq t < s\leq T$.
Hence, Eq.~\eqref{eq:reprV} reduces to
  \begin{align}\label{eq:simpl}
    \tilde{V}_t \ind_{\{\tau \geq t\}} & = \mathbb{E}^{\Qxx} \Bigl[ \Bigl.  \bigl(B_T^{r_f}\bigr)^{-1} \Phi(S_T) \ind_{\{\tau = T\}} + \bigl(B_\tau^{r_f}\bigr)^{-1} \theta_{\tau}(\hat{V})   \ind_{\{t < \tau <T\}} \Bigr. \Bigr. \, \Bigr\vert \, \mathcal{G}_t \Bigr] \nonumber \\
    & = \mathbb{E}^{\Qxx} \Bigl[ \Bigl. \bigl(B_T^{r_f}\bigr)^{-1} \Phi(S_T) - \bigl(B_T^{r_f}\bigr)^{-1} \Phi(S_T)  \ind_{\{t < \tau <T\}} + \bigl(B_\tau^{r_f}\bigr)^{-1}\theta_{\tau}(\hat{V})  \ind_{\{t < \tau <T\}} \Bigr. \Bigr. \, \Bigr\vert \, \mathcal{G}_t \Bigr] \nonumber \\
    & = \ind_{\{\tau \geq t\}} (B_t^{r_f})^{-1} \hat{V}(t,S_t) - \mathbb{E}^{\Qxx} \Bigl[\bigl(B_T^{r_f}\bigr)^{-1} \Phi(S_T)  \ind_{\{t < \tau <T\}} \Bigr\vert \,  \mathcal{G}_t \Bigr]  \nonumber \\
& \; \; +  \mathbb{E}^{\Qxx} \Bigl[\bigl(B_\tau^{r_f}\bigr)^{-1}
    \hat{V}(\tau,S_{\tau})  \ind_{\{t < \tau <T\}} \Bigr\vert \, \mathcal{G}_t \Bigr]
 \nonumber \\
& \; \;
 +  \mathbb{E}^{\Qxx} \Bigl[ \ind_{\{t < \tau <T\}} \bigl(B_\tau^{r_f}\bigr)^{-1}  \ind_{\{\tau_C<\tau_I \}} L_C (-(\hat{V}(\tau,S_{\tau}) - C_{\tau-}))^+\Bigr\vert \, \mathcal{G}_t\Bigr]
\nonumber  \\
    & \; \; - \mathbb{E}^{\Qxx} \Bigl[ \ind_{\{t < \tau <T\}} \bigl(B_\tau^{r_f}\bigr)^{-1}  \ind_{\{\tau_I<\tau_C \}} L_I  ((\hat{V}(\tau,S_{\tau}) - C_{\tau-}))^+   \Bigr\vert \, \mathcal{G}_t \Bigr].
  \end{align}
Above, the last equality is obtained using the definition of $\hat{V}$ given in Eq.~\eqref{eq:rulecoll} and the expression for the closeout amount given in Eq.\eqref{eq:theta}. Next, let us analyze the third term in the above decomposition. Using again the definition of $\hat{V}$, we obtain from the law of iterated expectations that
  \begin{align*}
    \mathbb{E}^{\Qxx} \Bigl[\bigl(B_\tau^{r_f}\bigr)^{-1} \hat{V}(\tau,S_{\tau})  \ind_{\{t < \tau <T\}} \Bigr\vert \, \mathcal{G}_t \Bigr] & = \mathbb{E}^{\Qxx} \Bigl[\bigl(B_{\tau}^{r_f}\bigr)^{-1} \mathbb{E}^{\Qxx} \Bigl[\bigl(B_T^{r_f} \bigr)^{-1} B_{\tau}^{r_f} \Phi(S_T) \Bigr\vert \, \mathcal{G}_{\tau} \Bigr] \ind_{\{t < \tau <T\}} \Bigr\vert \, \mathcal{G}_t \Bigr] \\
    & = \mathbb{E}^{\Qxx} \Bigl[\bigl(B_T^{r_f} \bigr)^{-1} \Phi(S_T) \ind_{\{t < \tau <T\}}  \Bigr\vert \, \mathcal{G}_t  \Bigr]
  \end{align*}
Plugging the above expression into~\eqref{eq:simpl} leads to
  \begin{align*}
    \tilde{V}_t \ind_{\{\tau \geq t\}} & = \ind_{\{\tau \geq t\}} (B_t^{r_f})^{-1} \hat{V}(t,S_t)   + \mathbb{E}^{\Qxx} \Bigl[ \ind_{\{t < \tau <T\}} \bigl(B_\tau^{r_f}\bigr)^{-1} \ind_{\{\tau_C<\tau_I \}} L_C (-(\hat{V}(\tau,S_{\tau}) - C_{\tau-}))^+   \Bigr\vert \, \mathcal{G}_t \Bigr]
 \\
    & - \mathbb{E}^{\Qxx} \Bigl[ \ind_{\{t < \tau <T\}} \bigl(B_\tau^{r_f}\bigr)^{-1} \ind_{\{\tau_I<\tau_C \}} L_I  ((\hat{V}(\tau,S_{\tau}) - C_{\tau-}))^+   \Bigr\vert \, \mathcal{G}_t \Bigr]
  \end{align*}
Multiplying both left and right hand side by $B_t^{r_f}$, we can rearrange the above expression and obtain
  \begin{align*}
    & \ind_{\{\tau \geq t\}} \bigl(V_t -\hat{V}(t,S_t)\bigr)  = \\
    & \phantom{=} \mathbb{E}^{\Qxx} \Bigl[ \ind_{\{t < \tau <T\}} \frac{B_t^{r_f}}{B_\tau^{r_f}} \Bigl( \ind_{\{\tau_C<\tau_I \}} L_C (-(\hat{V}(\tau,S_{\tau}) - C_{\tau-}))^+  -  \ind_{\{\tau_I<\tau_C \}} L_I  ((\hat{V}(\tau,S_{\tau}) - C_{\tau-}))^+ \Bigr)  \Bigr\vert \, \mathcal{G}_t \Bigr]  ,
  \end{align*}
which may be written as
  \begin{align*}
    \XVA_t = \DVA_t - \CVA_t,
  \end{align*}
where $\DVA$ and $\CVA$ are defined as
  \begin{align*}
    \DVA_t &:=  \mathbb{E}^{\Qxx} \biggl[\ind_{\{t < \tau <T\}} \ind_{\{\tau_I<\tau_C \}}\frac{B_t^{r_f}}{B_{\tau_I}^{r_f}}  L_I  (\hat{V}(\tau,S_{\tau}) - C_{\tau-})^+  \biggr\vert \, \mathcal{G}_t \biggr] \\
    \CVA_t &:= \mathbb{E}^{\Qxx} \biggl[\ind_{\{t < \tau <T\}} \ind_{\{\tau_C<\tau_I \}} \frac{B_t^{r_f}}{B_{\tau_C}^{r_f}} L_C (-(\hat{V}(\tau,S_{\tau}) - C_{\tau-}))^+ \biggr\vert \, \mathcal{G}_t \biggr],
  \end{align*}
and represent, respectively, the debit and credit valuation adjustments, see also equation 3.4 in \cite{capmig}.

\section{Conclusions} \label{sec:conclusions}
We have developed an arbitrage-free pricing framework for the total valuation adjustments (XVA) of a European claim that the trader purchases or sells to his counterparty. Our analysis takes into account funding spreads generated from the gap between borrowing and lending rates to the treasury, the repo market, collateral servicing costs, and counterparty credit risk.
The wealth process replicating payoffs of long and short positions in the traded claim can be characterized in terms of a nonlinear BSDE with random terminal condition,
associated with the closeout payment occurring when either party defaults. We have derived the no-arbitrage band associated with the prices of buyer and seller{'}s XVA, and
shown that it collapses to a unique XVA price only in the absence of rate asymmetries.
Such a setting corresponds to a generalization of \cite{Piterbarg}{'}s model for which we are able to derive an explicit expression of the XVA prices. This expression decomposes the
adjustment into four contributing components: funding costs of the default and collateral-free position, CVA, DVA, and servicing costs of the collateral procedure.

\appendix
\section{BSDEs -- Existence, Uniqueness, and Comparison of Solutions}\label{App_BSDE}

We prove here some results about BSDEs needed in the main body of the paper. We keep the general notation used throughout the paper and consider BSDEs of the form

  \begin{equation}\label{eq:genBSDE}
    \left\{ \begin{array}{rl}
    - dV_t & = f\bigl(\omega, t, V_t, Z_t, Z_t^I, Z_t^C\bigr) \, dt - Z_t \, dW_t^{\Qxx} - Z_t^I \, d\varpi_t^{I, \Qxx} - Z_t^C \, d\varpi_t^{C,\Qxx}\\
    V_\tau & = \zeta \ind_{\{\tau < T\}} + \vartheta \ind_{\{\tau = T\}} .
    \end{array} \right.
  \end{equation}
We call $(f,\zeta, \vartheta)$ the \textit{data} of the BSDE \eqref{eq:genBSDE}.

\begin{assumption}\label{BSDE_ass}
We will use the following assumption:
  \begin{enumerate}
    \item[\textbf{(A1)}] The terminal value satisfies $\zeta \in L^2\bigl(\Omega,\mathcal{F}_\tau, \Qxx\bigr)$ and  $\vartheta \in L^2\bigl(\Omega,\mathcal{F}_T, \Qxx\bigr)$.
    \item[\textbf{(A2)}] The generator $f  \, : \, \Omega \times [0,T] \times \R^4 \to \R$ is predictable and \textit{Lipschitz continuous} in $v$ $z$, $z^I$ and $z^C$, i.e. there is $K >0$ such that for all $(v_1, z_1, z^I, z^C)$, $(v_2, z_2, z^I, z^C) \in \R^4$ we have
    \[
      \bigl\vert f(\omega, t, v_1, z_1, z_1^I, z_1^C) - f(\omega, t, v_2, z_2, z_2^I, z_2^C)\bigr\vert \leq K \Bigl( \bigl\vert v_1 - v_2 \bigr\vert + \bigl\vert z_1 - z_2 \bigr\vert + \bigl\vert z^I_1 - z^C_2 \bigr\vert + \bigl\vert z^C_1 - z^C_2 \bigr\vert\Bigr)
    \]
    almost surely almost everywhere.
    \item[\textbf{(A3)}] $\E^{\Qxx} \Bigl[ \int_0^T \vert f(t, 0, 0, 0, 0) \vert^2 \, dt \Bigr] < \infty$.
    \item[\textbf{(A4)}] There is a constant $K >0$ such that for all  $(v, z, z_1^I, z_1^C)$, $(v, z, z_2^I, z_2^C) \in \R^4$ we have
    \[
      f(\omega, t, v, z, z_1^I, z^C) - f(\omega, t, v, z, z_2^I, z^C) \leq h_I^{\Qxx} \bigl(z_1^I - z_2^I\bigr)^- + K \bigl(z_1^I - z_2^I\bigr)^+
    \]
    and
    \[
      f(\omega, t, v, z, z^I, z_1^C) - f(\omega, t, v, z, z^I, z_2^C) \leq h_C^{\Qxx} \bigl(z_1^C - z_2^C\bigr)^- + K \bigl(z_1^C - z_2^C\bigr)^+
    \]
    almost surely almost everywhere.
  \end{enumerate}
\end{assumption}

To give an existence and uniqueness result for the BSDE \eqref{eq:genBSDE}, we introduce the following notation: Let $\mathbb{H}_t^2$ denote the space of predictable processes $Z \, : \, \Omega \times [0,t] \to \R$ satisfying
  \[
    \Exx^{\Qxx} \biggl[ \int_0^t \vert Z_s \vert ^2 \, ds \biggr] < \infty
  \]
and $\mathbb{S}_t^2$ the space of predictable processes $Z \, : \, \Omega \times [0,t] \to \R$ satisfying
  \[
    \Exx^{\Qxx} \biggl[ \sup_{s \in [0,t]} \vert Z_s \vert ^2 \biggr] < \infty.
  \]

\begin{theorem}\label{thm:main}
Assume that the data $(f,\zeta,\vartheta)$ satisfy the conditions (A1)--(A3) of \ref{BSDE_ass}. Then the BSDE \eqref{eq:genBSDE} admits a unique solution $(V^{\Qxx}, Z^{\Qxx}, Z^{I,\Qxx}, Z^{C,\Qxx}) \in \mathbb{S}_{\tau}^2 \times \mathbb{H}_{\tau}^2 \times \mathbb{H}_{\tau}^2 \times \mathbb{H}_{\tau}^2$.
\end{theorem}

\begin{proof}
We start writing the BSDE~\eqref{eq:genBSDE} driven by Cox processes into an equivalent BSDE driven by L\'{e}vy processes. Let $\tilde{N}^\Px_i(dt,dz)$, $i \in \{I,C\}$, be compensated Poisson random measures on $(\Omega, \mathcal{G}, \Px)$ with L\'{e}vy-compensators $h_i^{\Px}\, dt \delta_1(dz)$. Define the filtration $\mathbb{G}' := (\mathcal{G}'_t, t \geq 0)$, where
\[
\mathcal{G}'_t= \sigma\bigl(W_u^\Px, \tilde{N}^\Px_I([0,u],A), \tilde{N}^\Px_C([0,u],A); u \leq t, A \in \mathcal{B}(\R)\bigr).
\]
Note that by setting
  \[
    \tau_i'=\inf\{t\geq0;\ \tilde{N}^\Px_i([0,t], \R)=1\}, \qquad i\in\{I,C\},
  \]
and $\tau' = \tau_I' \wedge \tau_C' \wedge T$, we can identify $\mathbb{G}'_{\tau'}$ with $\mathbb{G}_\tau$. Moreover we can define the change of measure $\Qxx'$ via
\begin{equation}
\frac{d\Qxx'}{d\Px} \bigg|_{\mathcal{G}'_T} = e^{\frac{r_D-\mu}{\sigma}W_T^{\Px} - \frac{(r_D-\mu)^2}{2\sigma^2}T} e^{(2r_D-r^I-r^C)T} \prod_{i \in \{I,C\}} \prod_{0 \leq s \leq T}\biggl(1+\frac{r^i - r_D}{h_i^{\Px}} \ind_{\bigl\{\tilde{N}^\Px_i(s,\R) \neq \tilde{N}^\Px_i(s-,\R)\bigr\}}\biggr)
  \end{equation}
and note that $\Qxx'\vert_{\mathcal{G}'_{\tau'}} = \Qxx\vert_{\mathcal{G}_{\tau}}$. An application of Girsanov{'s} theorem gives that
$W^{\Qxx'}_t = W^{\Px} + \frac{\mu-r_D}{\sigma} t$, $t \geq 0$, is a $\Qxx'$ Brownian motion. Moreover, the compensated Poisson random measures  $\tilde{N}^{ \Qxx'}_i(dt,dz)$, $i \in \{I,C\}$, with L\'{e}vy-compensators $h_i^{\Qxx'}\, dt \delta_1(dz)$ are given via $h_i^{\Qxx'} = r^i - r_D + h_i^{\Px}$.

We define the function $\bar{f} \, : \, \Omega \times [0,T] \times \R^4 \to \R$
  \[
    \bar{f}(\omega, t, z, z^I, z^C) = f(\omega, t, z, z^I, z^C) \ind_{\{ t \leq \tau'\}}
  \]
and set $\bar{\vartheta} := \zeta \ind_{\{\tau' < T\}} + \vartheta \ind_{\{\tau' = T\}} \in L^2\bigl(\Omega,\mathcal{G}'_T, \tilde{\Px}'\bigr)$.
By a straightforward extension of the results of \cite[Section 3.1]{Delong} to two driving compensated random measures, the BSDE
  \begin{equation}\label{eq:barBSDE}
    \bar{V}_t = \bar{\vartheta} + \int_t^T \bar{f}(\omega, s, \bar{Z}_s, \bar{Z}_s^I, \bar{Z}_s^C) \, ds - \int_t^T \bar{Z}_s\, dW_s^{\Qxx'}  -\int_t^T \int_\R \bar{Z}_s^I \, \tilde{N}^{\Qxx'}_I(ds,dz)  - \int_t^T \int_\R \bar{Z}_s^C \, \tilde{N}^{\Qxx'}_C(ds,dz)
  \end{equation}
admits a unique solution in $\mathbb{S}_T^2 \times \mathbb{H}_T^2 \times \mathbb{H}_T^2 \times \mathbb{H}_T^2$. Generalizing an argument by Pardoux \cite[Proposition 2.6]{Pardoux}, this solution satisfies $\bar{V}_t = \bar{V}_{t \wedge \tau'}$ and $\bar{Z}_t \ind_{\{t > \tau'\}} = \bar{Z}_t^I \ind_{\{t > \tau'\}} = \bar{Z}_t^C \ind_{\{t > \tau'\}} =0$. Indeed, on the one hand we have (recall that $\tau' \leq T$ by definition)
  \[
    \bar{V}_{\tau'} = \bar{\vartheta} - \int_{\tau'}^T \bar{Z}_s\, dW_s^{\Qxx'}  -\int_{\tau'}^T \int_\R  \bar{Z}_s^I \, \tilde{N}^{\Qxx'}_I(ds,dz)  - \int_{\tau'}^T \int_\R \bar{Z}_s^C \, \tilde{N}^{\Qxx'}_C(ds,dz)
  \]
and therefore
  \[
    \bar{V}_{\tau'} = \Exx^{\tilde{\Qxx'}} \bigl[ \bar{\vartheta} \,\vert\, \mathcal{G}'_{\tau'} \bigr] = \bar{\vartheta}.
  \]
On the other hand, applying It\^{o}'s formula,
  \begin{align*}
    \bar{V}_{\tau'}^2 & = \bar{\vartheta}^2 - 2 \int_{\tau'}^T \bar{V}_s \bar{Z}_s\, dW_s^{\Qxx'}  - 2 \int_{\tau'}^T \int_\R \bar{V}_s \bar{Z}_s^I \, \tilde{N}^{\Qxx'}_I(ds,dz)  - 2 \int_{\tau'}^T \int_\R \bar{V}_s \bar{Z}_s^C \, \tilde{N}^{\Qxx'}_C(ds,dz) - \int_{\tau'}^T \bar{Z}_s^2 \, ds \\
    &\phantom{==} - \sum_{\tau' < s \leq T} \Bigl( \bar{V}_s^2 - \bar{V}_{s-}^2 - 2 \bar{V}_{s-} \bigl(\bar{V}_s - \bar{V}_{s-}\bigr)\Bigr)
  \end{align*}
Thus, taking conditional expectations and noting that the stochastic integrals are all true martingales as $(\bar{V}, \bar{Z}, \bar{Z}^I, \bar{Z}^C) \in \mathbb{S}_T^2 \times \mathbb{H}_T^2 \times \mathbb{H}_T^2 \times \mathbb{H}_T^2$, we get that
  \[
    \Exx^{\Qxx'}\biggr[ \Bigl. \int_{\tau'}^T \bar{Z}_s^2 \, ds + \sum_{\tau' < s \leq T} \bigl( \bar{V}_s - \bar{V}_{s-} \bigr)^2 \, \Bigr\vert \, \mathcal{G}'_{\tau'} \biggr] = 0
  \]
and therefore $\int_{\tau'}^T \bar{Z}_s^2 \, ds = 0$ a.s and $\bar{V}_t \ind_{\{t > \tau'\}} = \bar{V}_{\tau'} \ind_{\{t > \tau'\}}$ a.s. Therefore, as $f$ and $\bar{f}$ agree up to $\tau'$ and $\varpi_t^{i,\Qxx}$ and $\tilde{N}^{\Qxx'}_i([0,t],\R)$, $i \in \{I,C\}$, agree on $[0,\tau']$, the unique solution to \eqref{eq:barBSDE} restricted to $\{ t \leq \tau'\}$ is also a solution to \eqref{eq:genBSDE}. Moreover this solution is unique as every solution $(V^{\Qxx'}, Z^{\Qxx'}, Z^{I,\Qxx'}, Z^{C,\Qxx'})$ of \eqref{eq:genBSDE} can be extended to a solution of \eqref{eq:barBSDE} by setting $\bar{V}'_s = V^{\Qxx'}_{s \wedge \tau}$ and $\bar{Z}'_s = Z^{\Qxx'}_s \ind_{\{s \leq \tau\}}$, $\bar{Z}_s^{'I} = Z_s^{I,\Qxx'} \ind_{\{s \leq \tau\}}$, $\bar{Z}_s^{'C} = Z_s^{C,\Qxx'} \ind_{\{s \leq \tau\}}$.
\end{proof}

In a similar way, we can prove the following comparison theorem. Here, we only provide a sketch of the proof.

\begin{theorem}\label{thm:compare}
Let $(f_1,\zeta_1,\vartheta_1)$ and $(f_2, \zeta_2, \vartheta_2)$ be data of BSDEs satisfying the assumptions (A1) -- (A4) of \ref{BSDE_ass} and such that $f_1(t,v,z,z^I,z^C) \ge f_2(t,v,z,z^I,z^C)$ for all $(t,v,z,z^I,z^C) \in [0,T] \times \R^4$, $\zeta_1\ge \zeta_2$ a.s. and $\vartheta_1\ge\vartheta_2$ a.s. Then $V^1_{t \wedge \tau} \geq V^2_{t  \wedge \tau}$ a.s. for all $t \in [0,T]$. Moreover, if $V^1_{t_0 \wedge \tau} = V^2_{t_0  \wedge \tau}$ a.s. for some $t_0 \in [0,T]$, then $V^1_{t \wedge \tau} = V^2_{t \wedge \tau}$ a.s for all $t \in [t_0,T]$.
\end{theorem}

\begin{proof}
The proof works similar to the one given above. We consider first the BSDEs with drivers $ \bar{f}_j(\omega, t, z, z^I, z^C) = f_j(\omega, t, z, z^I, z^C) \ind_{\{ t \leq \tau\}}$, $ j \in \{1,2\}$, and driven by compensated Poisson random measures $\tilde{N}^{\Qxx'}_i$, $i \in \{I,C\}$.
We can thus invoke the comparison theorem \cite[Theorem 3.2.2]{Delong}. Restricting ourselves to  $\{ t \leq \tau' \}$ (where the Poisson random measures agree with the martingales $\varpi^{i,\Qxx}$), we get the result.
\end{proof}

We next make the following observation.

\begin{remark}\label{remark:PDE1}
Similar to Proposition 4.1.1 in \cite{Delong} we know that the value process up to default is Markovian. Moreover, on the set $\{t<\tau\}$ we have that  $\varpi^{i,\Qxx}_t = -h_i^\Qxx t$ for $i \in \{I,C\}$. Thus we can reduce it to the state variables $t$ and $S_t$. Therefore we can find a measurable function $\vv \, : \, [0,T] \times \mathbb{R}_{>0} \to \R$ such that $\vv(t, S_t) =V_t \ind_{\{\tau>t\}}$.
\end{remark}

The following theorem shows how the function $\vv$ from Remark \ref{remark:PDE1} can be used to find the hedging strategies. Once noting that the law of $S_t$ is absolutely continuous, the proof of the theorem becomes analogous to the proof of Theorem 4.1.4 in \cite{Delong} and will be omitted.
\begin{theorem}\label{thm:strat}
Let $(f,\theta_\tau(\hat V(\tau, S_\tau)),\hat V(T, S_T))$ be data of BSDEs satisfying the assumptions (A1) -- (A4) of \ref{BSDE_ass}. Additionally, let the function $\vv(t, S_t) =V_t \ind_{\{\tau>t\}}$ be defined as in
Remark \ref{remark:PDE1}. Then, on the set $\{t<\tau\}$ we have that
  \begin{align*}
    Z_t &= \sigma S_t \vv_S(t,S_t),\\
     Z_t^i &= \theta_i(\hat V(t, S_t)) - {\vv(t,S_{t})}, \qquad \qquad i \in \{I,C\}.
  \end{align*}
\end{theorem}

\end{document}